\documentclass{article}
\usepackage[utf8]{inputenc}
\usepackage[total={6.2in, 8in}]{geometry}

\usepackage{amsmath}
\usepackage{amssymb}
\usepackage{amsthm}
\usepackage{appendix}
\usepackage{graphicx}
\usepackage{algorithmic}
\usepackage{algorithm}
\usepackage[roman]{complexity}
\usepackage{mathrsfs}
\usepackage{braket}
\usepackage{xcolor}
\usepackage[sort,square,numbers]{natbib}
\usepackage{doi}
\usepackage{tikz}
\usetikzlibrary{quantikz}
\usepackage{authblk}

\newtheorem{thm}{\protect\theoremname}
\theoremstyle{plain}
\newtheorem{lem}[thm]{\protect\lemmaname}
\theoremstyle{plain}

\theoremstyle{plain}
\newtheorem*{lem*}{\protect\lemmaname}
\theoremstyle{plain}

\theoremstyle{plain}
\newtheorem{cor}[thm]{\protect\corollaryname}

\newtheorem{defn}[thm]{Definition}

\newcommand{\ketbra}[2]{\lvert #1 \rangle \! \langle #2 \rvert}
\newcommand{\norm}[1]{\left\lVert#1\right\rVert}

\usepackage[USenglish]{babel}
  \providecommand{\corollaryname}{Corollary}
  \providecommand{\lemmaname}{Lemma}
  \providecommand{\propositionname}{Proposition}
  \providecommand{\remarkname}{Remark}
\providecommand{\theoremname}{Theorem}

\usepackage[normalem]{ulem}

\DeclareMathOperator{\Tr}{tr}
\DeclareMathOperator*{\Exp}{{\mathbb{E}}}
\newcommand{\Or}{\mathcal{O}}

\newcommand{\tr}{\mathrm{tr}}

\newcommand{\ad}{\mathrm{ad}}

\newcommand\CoAuthorMark{\footnotemark[\arabic{footnote}]}

\title{Learning many-body Hamiltonians with\\ Heisenberg-limited scaling}

\author[1]{Hsin-Yuan Huang \footnote{These authors contributed equally to this work.}}
\author[1,2]{Yu Tong \protect\CoAuthorMark}
\author[2,3]{Di Fang}
\author[4]{Yuan Su}
\affil[1]{Institute for Quantum Information and Matter, California Institute of Technology}
\affil[2]{Department of Mathematics, University of California, Berkeley}
\affil[3]{Simons Institute for the Theory of Computing, University of California, Berkeley}
\affil[4]{Microsoft Quantum}

\begin{document}

\maketitle

\begin{abstract}
Learning a many-body Hamiltonian from its dynamics is a fundamental problem in physics. In this work, we propose the first algorithm to achieve the Heisenberg limit for learning an interacting $N$-qubit local Hamiltonian. After a total evolution time of~$\mathcal{O}(\epsilon^{-1})$, the proposed algorithm can efficiently estimate any parameter in the $N$-qubit Hamiltonian to $\epsilon$-error with high probability. The proposed algorithm is robust against state preparation and measurement error, does not require eigenstates or thermal states, and only uses $\mathrm{polylog}(\epsilon^{-1})$ experiments. In contrast, the best previous algorithms, such as recent works using gradient-based optimization or polynomial interpolation, require a total evolution time of $\mathcal{O}(\epsilon^{-2})$ and $\mathcal{O}(\epsilon^{-2})$ experiments. Our algorithm uses ideas from quantum simulation to decouple the unknown $N$-qubit Hamiltonian~$H$ into noninteracting patches, and learns $H$ using a quantum-enhanced divide-and-conquer approach. We prove a matching lower bound to establish the asymptotic optimality of our algorithm.
\end{abstract}

\clearpage
\tableofcontents
\clearpage

\section{Introduction}

Learning an unknown Hamiltonian $H$ from its dynamics $U(t) = e^{-iHt}$ is an important problem that arises in quantum sensing/metrology \cite{de2005quantum,valencia2004distant,leibfried2004toward,bollinger1996optimal,lee2002quantum,mckenzie2002experimental,holland1993interferometric,wineland1992spin,caves1981quantum}, quantum device engineering \cite{boulant2003,innocenti2020,ben2020,shulman2014,sheldon2016,sundaresan2020,zhang2022scalable}, and quantum many-body physics \cite{wiebe2014a,wiebe2014b,verdon2019,burgarth2017, wang2017, kwon2020, wang2020, cotler2021emergent, choi2021emergent, huang2020predicting}.
In quantum sensing/metrology, the Hamiltonian $H$ encodes signals that we want to capture.
A more efficient method to learn $H$ implies the ability to extract these signals faster, which could lead to substantial improvement in many applications, such as microscopy, magnetic field sensors, positioning systems, etc.
In quantum computing, learning the unknown Hamiltonian $H$ is crucial for calibrating and engineering the quantum device to design quantum computers with a lower error rate.
In quantum many-body physics, the unknown Hamiltonian $H$ characterizes the physical system of interest.
Obtaining knowledge of $H$ is hence crucial to understanding microscopic physics.
A central goal in these applications is to find the most efficient approach to learning $H$.

In this work, we focus on the task of learning many-body Hamiltonians describing a quantum system with a large number of constituents.
For concreteness, we consider an $N$-qubit system.
Given any unknown $N$-qubit Hamiltonian $H$, we can represent $H$ in the following form,
\begin{equation}
    H = \sum_{E \in \{I, X, Y, Z\}^{\otimes N}} \lambda_E E,
\end{equation}
where $\lambda_E \in \mathbb{R}$ are the unknown parameters.
The goal of learning the unknown Hamiltonian $H$ is hence equivalent to learning $\lambda_E$ for each $N$-qubit Pauli operator $E$.
In previous works on learning many-body Hamiltonians \cite{HaahKothariTang2021optimal,yu2022,hangleiter2021,FrancaMarkovichEtAl2022efficient,ZubidaYitzhakiEtAl2021optimal,BaireyAradEtAl2019learning,GranadeFerrieWiebeCory2012robust,gu2022practical,wilde2022learnH,KrastanovZhouEtAl2019stochastic}, in order to reach an $\epsilon$ precision in estimating the parameters $\lambda_E$, the number of experiments and the total time required to evolve the system have a scaling of at least $\epsilon^{-2}$.
However, the $\epsilon^{-2}$ precision scaling is likely not the best-possible scaling for learning an unknown many-body Hamiltonian $H$ from dynamics.

In quantum sensing/metrology, the scaling of $\epsilon^{-2}$ for learning an unknown parameter to $\epsilon$ error is known as the standard quantum limit.
For simple classes of Hamiltonians, such as when $H$ contains only one unknown parameter or when $H$ describes a single-qubit system, one can surpass the standard quantum limit using quantum-enhanced protocols \cite{giovannetti2011advances,Zhou2017AchievingTH, degen2017quantum,holland1993interferometric,leibfried2004toward,de2005quantum}.
The true limit set by the basic principles of quantum mechanics is known as the Heisenberg limit, which gives a scaling of $\epsilon^{-1}$.
Assuming quantum mechanics is true, the Heisenberg limit states that the scaling of the \textit{total evolution time} must be at least of order $\epsilon^{-1}$.
If a protocol uses $J$ experiments, where the $j$-th experiment uses the unknown Hamiltonian evolution $e^{-iH t_{j, 1}}, \ldots, e^{-iH t_{j, K_j}}$ for some time $t_{j, 1}, \ldots, t_{j, K_j}$, then the total evolution time is defined as
\begin{equation}
    T \triangleq \sum_{j=1}^J \sum_{k=1}^{K_j} t_{j, k}.
\end{equation}
Other measures of complexity, e.g., the number of experiments, could surpass the $\epsilon^{-1}$ precision scaling, but that does not imply that the Heisenberg limit is beaten \cite{giovannetti2011advances, degen2017quantum}.

There are two well-established quantum-enhanced approaches for achieving the Heisenberg limit in learning simple Hamiltonians, such as a single-qubit Hamiltonian $H = \omega Z$ with unknown parameter $\omega$.
The first approach \cite{lee2002quantum,bollinger1996optimal,leibfried2004toward} considers evolving a highly-entangled state over $\ell = \mathcal{O}(\epsilon^{-1})$ copies of the system under $\ell$ copies of the unknown Hamiltonian dynamics $U(t)^{\otimes \ell}$.
The second approach \cite{de2005quantum,higgins2007entanglement,KimmelLowYoder2015robust} considers long-time coherent evolution with time $t = \mathcal{O}(\epsilon^{-1})$ over a single copy of the system.
However, both approaches are challenging to apply in many-body systems with a large system size $N$ and many unknown parameters.
The difficulty stems from the many-body interactions in the Hamiltonian $H$.
As time $t$ becomes larger, the entanglement growth in $e^{-i t H}$ will cause all the unknown parameters in $H$ to tangle with one another.
Furthermore, the many-body entanglement can be seen as a form of decoherence,
which kills the quantum enhancement.
To prevent the system from becoming too entangled, prior work on learning many-body Hamiltonians focuses on a short time $t$, which loses the quantum enhancement and obtains at best the standard quantum limit scaling as $\epsilon^{-2}$.

\begin{figure*}[t]
\centering
\includegraphics[width=0.98\textwidth]{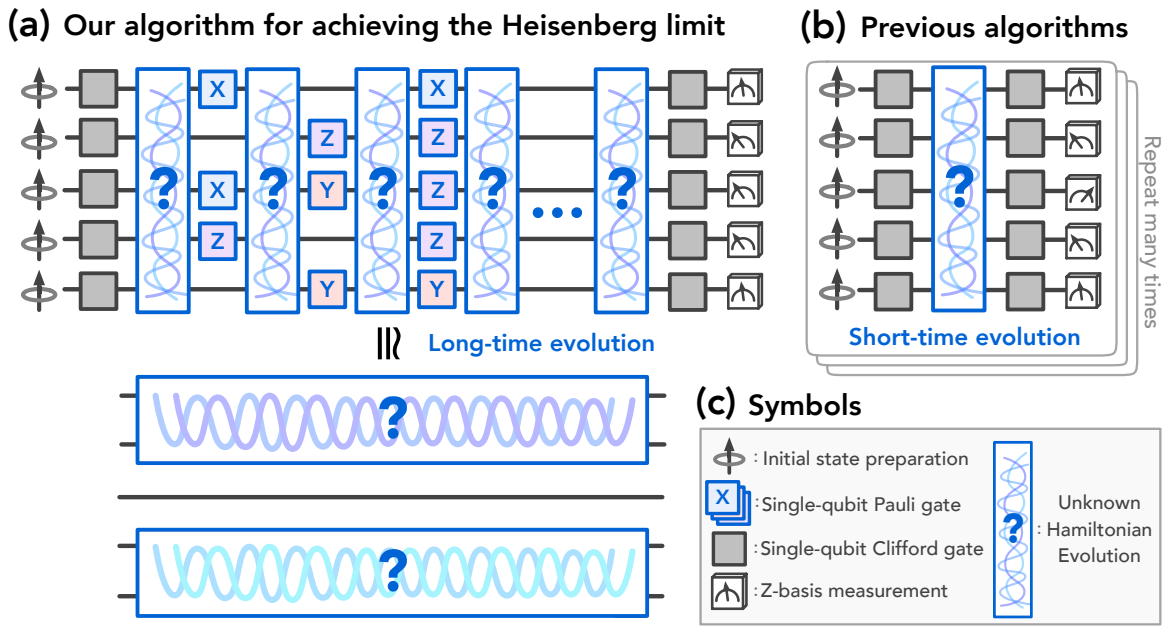}
    \caption{
    Algorithms for learning many-body Hamiltonians.
    \emph{(a)~Our algorithm for achieving the Heisenberg limit~$\epsilon^{-1}$:} We use ideas from quantum simulation to decouple the unknown Hamiltonian $H$ into non-interacting patches, where the unknown local Hamiltonian on each patch has known eigenvectors.
    We then perform long-time coherent evolution for each non-interacting patch to learn the unknown parameters.
    One only needs $\mathcal{O}(\mathrm{polylog}(\epsilon^{-1}))$ experiments and a total evolution time of $\mathcal{O}(\epsilon^{-1})$.
    \emph{(b)~Previous algorithms for achieving the standard quantum limit~$\epsilon^{-2}$:} Previous methods \cite{HaahKothariTang2021optimal,FrancaMarkovichEtAl2022efficient,gu2022practical,wilde2022learnH} repeatedly run a short-time evolution under the unknown Hamiltonian $H$ incoherently for many times. One needs $\mathcal{O}(\epsilon^{-2})$ experiments and a total evolution time of $\mathcal{O}(\epsilon^{-2})$.
    \emph{(c) Symbols:} The symbols used in (a, b). The unknown Hamiltonian evolution is $U(t) = e^{-i t H}$.
    \label{fig:LearnH}}
\end{figure*}

In this paper, we propose the first learning algorithm to achieve the Heisenberg limit for learning interacting many-body Hamiltonian.
We prove that the proposed algorithm can learn a model of an unknown $N$-qubit local Hamiltonian $H$ after a total evolution time of
\begin{equation}
    T = \mathcal{O}\left( \epsilon^{-1} \log(\delta^{-1}) \right)
\end{equation}
which is independent of the system size $N$, such that for any parameter in the unknown $N$-qubit Hamiltonian~$H$, the algorithm can estimate the parameter to at most $\epsilon$ error with probability at least $1 - \delta$.
The proposed algorithm only uses $\mathcal{O}\left(\mathrm{polylog}\left(\epsilon^{-1}\right) \log\left(\delta^{-1}\right) \right)$ experiments.
Furthermore, after running the experiments, the classical computational time of the proposed learning algorithm to estimate all parameters only needs to be of $\mathcal{O}(N \mathrm{polylog}(\epsilon^{-1}) \log(\delta^{-1}))$.
In quantum sensing/metrology, the failure probability $\delta$ is usually considered to be a fixed constant, e.g., $0.01$.
In this setting, our algorithm achieves a scaling of $\mathcal{O}(\epsilon^{-1})$ saturating the Heisenberg limit.

The proposed algorithm is robust against state preparation and measurement (SPAM) errors.
To establish the optimality of the proposed algorithm, we prove a matching lower bound of
\begin{equation}
    T = \Omega\left( \epsilon^{-1} \log(\delta^{-1}) \right)
\end{equation}
for any learning algorithm robust against SPAM error.
The lower bound can be seen as an algorithmic proof of the Heisenberg limit with the failure probability $\delta$ taken into account.

Our learning algorithm has the additional advantage of using only single-qubit Clifford gates, and not requiring eigenstates or thermal states of the Hamiltonian $H$.
The shortcomings are that the total evolution time $T$ would have an explicit dependence on $N$ to achieve $\epsilon^{-1}$ scaling in learning quantum systems with all-to-all long-range interactions,
and that the precision $\epsilon$ it can achieve is limited by how fast we can apply single-qubit Pauli gates.

\section{Main results}

Given a system size $N$.
We focus on learning an unknown $N$-qubit Hamiltonian 
$H = \sum_a \lambda_a E_a$ that can be written as a linear combination of few-body terms $E_a$, where each qubit is acted on by $\mathcal{O}(1)$ of the few-body terms.
Such a Hamiltonian can always be written as
\begin{equation} \label{eq:low-int-H}
    H = \sum_{a=1}^M \lambda_a E_a,
\end{equation}
where $\lambda_1, \ldots, \lambda_M$ are the unknown parameters and $S = \{E_1, \ldots, E_M\} \subseteq \{I, X, Y, Z\}^{\otimes N}$ is a subset of $N$-qubit Pauli operators.
Each Pauli operator $E_a$ acts nontrivially on $k = \mathcal{O}(1)$ qubits and each qubit is acted on by $\mathcal{O}(1)$ of the Pauli operators in $S$.
The number of unknown parameters is equal to $M = |S| = \Theta(N)$.
We refer to this class of Hamiltonians as \textit{low-interaction Hamiltonians} following \cite{HaahKothariTang2021optimal, gu2022practical}. 
This class of Hamiltonians includes geometrically-local Hamiltonians as a special case and is also referred to as bounded-degree local Hamiltonians \cite{anshu2021improved, harrow2017extremal} in the literature. 
Following \cite{HaahKothariTang2021optimal, gu2022practical}, we assume that $S$ is fixed and known.

We consider algorithms that can learn from experiments involving the unknown $N$-qubit Hamiltonian dynamics $U(t) = e^{-i H t}$.
Each experiment prepares an initial state with an arbitrary number of ancillas, evolves under an interleaving sequence of unknown Hamiltonian dynamics and controllable quantum circuit,
\begin{equation}
V_{K+1} \, U(t_K) \, V_K \ldots V_2 \, U(t_1) \, V_1    
\end{equation}
where $K$ is some integer, $t_1, \ldots, t_K$ are the evolution times, and $V_1, \ldots, V_{K+1}$ are the controllable circuits,
and ends with a POVM measurement.
This is similar to definitions considered in \cite{aharonov2022quantum, chen2022exponential, huang2022foundations, huang2022quantum}.
To model SPAM error, we assume that the actual initial state and the actual POVM implemented are only approximately equal to the ideal initial state and ideal POVM.

We consider a simple set of experiments, where the initial state is a noisy all-zero state $\ket{0^N}$, each controllable circuit $V_k$ is a layer of single-qubit Clifford gates, and the POVM is a noisy computational basis measurement.
We refer to these experiments as single-qubit Clifford experiments.
We give a learning algorithm with a rigorous upper bound on the total evolution time, as stated in the theorem below (more detailed statements can be found in Theorems \ref{thm:ham_learn_partial_pauli_twirling} and \ref{thm:ham_learn_Trotter} (in Sections \ref{sec:est_all_bases_clusters} and \ref{sec:diagonal_trotter} respectively) on the number of Clifford gates and experiments needed).

\begin{thm} \label{thm:main-up}
There is a learning algorithm robust to SPAM error and restricted to single-qubit Clifford experiments that achieves the following.
For any unknown $N$-qubit low-interaction Hamiltonian $H = \sum_{a=1}^M \lambda_a E_a$ with $|\lambda_a| \leq 1$, after a total evolution time $T = \mathcal{O}(\epsilon^{-1} \log(\delta^{-1}))$,
the learning algorithm can obtain estimates $\hat{\lambda}_a$ from the experiments, such that
\begin{equation}
    \Pr\left[ \left| \hat{\lambda}_a - \lambda_a \right| \leq \epsilon \right] \geq 1 - \delta,
\end{equation}
for all $a \in \{1, \ldots, M\}$.
The classical computational time to generate all the estimates $\hat{\lambda}_a, \forall a \in \{1, \ldots, M\}$ from experimental data is $\mathcal{O}(N \mathrm{polylog}(\epsilon^{-1}) \log(\delta^{-1}))$.
\end{thm}

We also prove the following matching lower bound for any learning algorithm that can execute arbitrary quantum experiments involving the unknown Hamiltonian dynamics $U(t) = e^{-i H t}$ adaptively based on previous experiments.

\begin{thm} \label{thm:main-down}
Suppose there is a learning algorithm robust to SPAM error that achieves the following.
For any unknown $N$-qubit low-interaction Hamiltonian $H = \sum_{a=1}^M \lambda_a E_a$ with $|\lambda_a| \leq 1$, after a total evolution time~$T$, the learning algorithm can obtain estimates $\hat{\lambda}_a$ from the experiments, such that
\begin{equation}
    \Pr\left[ \left| \hat{\lambda}_a - \lambda_a \right| \leq \epsilon \right] \geq 1 - \delta,
\end{equation}
for all $a \in \{1, \ldots, M\}$.
Then, we have $T = \Omega(\epsilon^{-1} \log(\delta^{-1}))$.
\end{thm}

While we focus on the complexity to estimate \emph{any individual} parameter to $\epsilon$ error with probability at least $1 - \delta$,
one could also consider the estimation of \emph{all} parameters to $\epsilon$ error with probability at least $1 - \delta'$.
By choosing $\delta = \delta' / N$ and using union bound, the former guarantee implies the latter.
Hence, our proposed algorithm has the guarantee that after a total evolution time of $T = \mathcal{O}(\epsilon^{-1} \log(N / \delta'))$,
it can estimate all parameters in the unknown $N$-qubit Hamiltonian $H$ to $\epsilon$ error with probability at least $1 - \delta'$.

We also show that the proposed algorithm can learn from very few experiments.
The number of experiments only needs to be of
\begin{equation}
    \mathcal{O}\left(\mathrm{polylog}\left(\epsilon^{-1}\right) \log\left(\delta^{-1}\right) \right),
\end{equation}
which is significantly lower than $\Theta(\epsilon^{-1})$.
This does not mean that we have surpassed the Heisenberg limit since the limit is concerned with the total evolution time.
In addition to the total evolution time and the number of experiments, we also show that the proposed algorithm uses $\mathcal{O}(\epsilon^{-1.5} \mathrm{polylog}(\epsilon^{-1}) \log(\delta^{-1}))$ layers of single-qubit Clifford gates, which results in a total of $\mathcal{O}(N \epsilon^{-1.5} \mathrm{polylog}(\epsilon^{-1}) \log(\delta^{-1}))$ single-qubit Clifford gates.

\section{Proof ideas}
\label{sec:overview}

In this section, we provide ideas for designing the proposed learning algorithm and establishing the proof of the main results.
All parts except for the last are devoted to Theorem~\ref{thm:main-up} on an efficient learning algorithm.
The last part is on the lower bound given in Theorem~\ref{thm:main-down}.

\subsection{Reshaping an unknown Hamiltonian}

A key technique used throughout the design of the learning algorithm is the idea of reshaping an unknown Hamiltonian using Hamiltonian simulation techniques.
Recall that given a set of Hamiltonians $H_1, \ldots, H_K$ and the ability to implement the unitaries $e^{-i t H_1}, \ldots, e^{-i t H_K}$,
many Hamiltonian simulation techniques allow one to implement the unitary approximately
\begin{equation}
    e^{-i t \sum_{k=1}^K H_k}.
\end{equation}
Note that these approximation formulas are valid for unitaries, and no knowledge of the underlying Hamiltonian is required.
As such, they apply to the learning problem considered here.

For example, a randomized Hamiltonian simulation algorithm known as qDRIFT \cite{Campbell2019random,BerryChildsEtAl2020time,chen2021concentration} considers an approximation (as a quantum channel) given by
\begin{equation}
    e^{-i t \sum_{k=1}^K H_k} \approx e^{-i (t/r) H_{k_r}} \ldots e^{-i (t/r) H_{k_1}},
\end{equation}
where $r$ is an integer that sets the approximation error, $k_1, \ldots, k_r$ are independent random variables sampled according to some probability distribution over $\{1, \ldots, K\}$. 
On the other hand, the first-order and second-order Trotterization method \cite{Suzuki91,Lloyd96,BACS05} considers approximations given by
\begin{equation}
\begin{aligned}
    e^{-i t \sum_{k=1}^K H_k} &\approx \left(e^{-i (t/r) H_K} \ldots e^{-i (t/r) H_1}\right)^{r},\\
    e^{-i t \sum_{k=1}^K H_k} &\approx \left(e^{-i (t/2r) H_1} \ldots e^{-i (t/2r) H_K}e^{-i (t/2r) H_K} \ldots e^{-i (t/2r) H_1}\right)^{r}.
\end{aligned}
\end{equation}
Although higher-order Trotterization formulas can perform better than the above first- and second-order formulas \cite{childs2019theory}, they nevertheless require the system to evolve backward in time and are thus not suitable in the Hamiltonian learning setting.
As $r \rightarrow \infty$, these approximations become exact.

Now, consider the unknown $N$-qubit Hamiltonian $H$ that we hope to learn.
Given any $N$-qubit unitaries $U_1, \ldots, U_K$ and some associated weights $w_1, \ldots, w_K \geq 0$.
We define a new unknown $N$-qubit Hamiltonian as follows,
\begin{equation} \label{eq:reshapeH}
    \widetilde{H} \triangleq \sum_{k=1}^K w_k U_k H U_k^\dagger.
\end{equation}
Let us consider $H_k \triangleq w_k U_k H U_k^\dagger$ for all $k \in \{1, \ldots, K\}$.
A standard identity for matrix exponential implies
\begin{equation}
    e^{-i t H_K} = U_k e^{-i (w_k t) H} U_k^\dagger.
\end{equation}
Hence, we can implement $e^{-i t H_K}$ using unitary dynamics $U(t) = e^{-i t H}$ under the unknown Hamiltonian $H$.
Then using Hamiltonian simulation techniques, we can evolve under the $N$-qubit unitary $e^{-i t \widetilde{H}}$.
In the special case where $U_k$ are powers of the same unitary, this technique approximately projects $H$ into the quantum Zeno subspace determined by the generator \cite{Tran21,Burgarth2022oneboundtorulethem}.
However, we use the general version of this technique which allows us to reshape any unknown $N$-qubit Hamiltonian $H$ to another Hamiltonian $\widetilde{H}$ given by Eq.~\eqref{eq:reshapeH}, and evolve under the new unknown Hamiltonian $\widetilde{H}$.
The reshaping will lead to a small approximation error, which we discuss in Section~\ref{sec:char-reshape}.

\subsection{Learning a single-qubit Hamiltonian}

We now show how the Hamiltonian reshaping technique can be very useful in learning Hamiltonians.
We begin with a simple question: how to learn a single-qubit Hamiltonian with Heisenberg-limited precision scaling?
If we naively apply quantum process tomography \cite{mohseni2008quantum, Scott08, o2004quantum, levy2021classical, huang2022foundations, merkel2013self, nielsen2020gate, blume2017demonstration} to learn the unknown Hamiltonian,
we would have an $\epsilon^{-2}$ dependence in the number of measurements needed, where $\epsilon$ is the desired precision of the Hamiltonian parameters.
Therefore we need to consider a different method.
Any single-qubit Hamiltonian can be written as 
\begin{equation}
    \label{eq:single_qubit_ham}
    H = \lambda_x X + \lambda_y Y + \lambda_z Z,
\end{equation}
where $|\lambda_x|,|\lambda_y|,|\lambda_z|\leq 1$. We want to learn each parameter with additive error at most $\epsilon$. We also want to have high confidence in the estimate we get, and to this end, we require that for each estimate, the probability of having an error larger than $\epsilon$ is at most $\delta$.

The above problem would become easy if we knew a priori that, for example, $\lambda_y=\lambda_z=0$.
In this case, the unknown Hamiltonian is given by $H = \lambda_x X$, which is a standard setup considered in quantum metrology \cite{de2005quantum,valencia2004distant,leibfried2004toward,bollinger1996optimal,lee2002quantum,mckenzie2002experimental,holland1993interferometric,wineland1992spin,caves1981quantum}.
There are many approaches to achieving the Heisenberg limit for this very simple class of Hamiltonians, including those based on highly entangled states \cite{lee2002quantum,bollinger1996optimal,leibfried2004toward} and those based on long-time evolution \cite{de2005quantum,higgins2007entanglement,KimmelLowYoder2015robust}.
Here, we consider an approach based on long-time evolution, known as robust phase estimation \cite{KimmelLowYoder2015robust}.
This approach can estimate $\lambda_x$ to $\epsilon$ accuracy with probability at least $1 - \delta$ with a total evolution time scaling like $\Or(\epsilon^{-1}\log(\delta^{-1}))$.
This approach relies crucially on the knowledge of the eigenstates of $H = \lambda_x X$,
which is unavailable for a general unknown single-qubit Hamiltonian given in \eqref{eq:single_qubit_ham}.
For a general single-qubit Hamiltonian, we would not know the eigenstates unless we first learn the parameters of the Hamiltonian.

We resolve this dilemma by eliminating the unwanted terms using the technique of reshaping an unknown Hamiltonian. With the unwanted terms removed, we can focus on the term we want to estimate.
Let us first consider the estimation of $\lambda_x$, and we want to keep $Y$ and $Z$ from interfering with our estimation.
To achieve this, we consider reshaping the Hamiltonian $H$ using $U_1 = I, U_2 = X$ and $w_1 = w_2 = \tfrac{1}{2}$. The new unknown Hamiltonian is given by
\begin{equation}
    \widetilde{H} \triangleq \frac{1}{2} H + \frac{1}{2} X H  X = \lambda_x X.
\end{equation}
Here we have used the fact that $XYX=-Y$ and $XZX=-Z$. Note that this effective Hamiltonian is exactly what we want! With this Hamiltonian, we can directly apply the robust phase estimation algorithm in \cite{KimmelLowYoder2015robust} to estimate $\lambda_x$. The same thing can be done for $\lambda_y$ and $\lambda_z$ as well. This enables us to estimate the parameters of a single-qubit Hamiltonian with $\Or(\epsilon^{-1}\log(\delta^{-1}))$ total evolution time, and $\Or(\polylog(\epsilon^{-1})\log(\delta^{-1}))$ number of experiments.

\subsection{Learning a few-qubit Hamiltonian}

We can generalize the above idea for learning a single-qubit Hamiltonian to a few-qubit Hamiltonian.
For a Hamiltonian acting on $\Or(1)$ qubits, we can learn all the parameters involved using $\Or(\epsilon^{-1}\log(\delta^{-1}))$ total evolution time, and $\Or(\polylog(\epsilon^{-1})\log(\delta^{-1}))$ number of experiments.
As an example, let us consider an arbitrary two-qubit Hamiltonian 
\begin{equation}
    H = \sum_{P, P' \in \{I,X,Y,Z\}}\lambda_{P P'} P_1 P'_2
\end{equation}
with $|\lambda_{P P'}|\leq 1$. Here $P_1$ and $P'_2$ denote the Pauli gates $P$ and $P'$ acting on qubits $1$ and $2$ respectively. Suppose we want to estimate the parameter $\lambda_{X Z}$. Then we can consider reshaping the unknown Hamiltonian $H$ using $U_1 = I, U_2 = X_1, U_3 = Z_2, U_4 = X_1Z_2$ and $w_1 = w_2 = w_3 = w_4 = \tfrac{1}{4}$.
The new unknown Hamiltonian after reshaping is given by
\begin{equation}
    \widetilde{H} \triangleq \frac{1}{4}(H + X_1 H X_1 + Z_2HZ_2 + X_1Z_2 HX_1Z_2) = \lambda_{XZ} X_1 Z_2 + \lambda_{XI} X_1 + \lambda_{IZ} Z_2.
\end{equation}
This is because the averaging eliminates all Pauli terms in $H$ that do not have $I$ or $X$ on the first qubit and $I$ or $Z$ on the second qubit.

This new unknown Hamiltonian $\widetilde{H}$ after the reshaping is not as simple as the one we get in the single-qubit case.
However, we still have access to its eigenstates.
This is because, in this new Hamiltonian $\widetilde{H}$, only one (non-identity) Pauli operator is associated with each qubit.
The eigenbasis for the new unknown Hamiltonian $\widetilde{H}$ is always given by
$\{\ket{+}\ket{0},\ket{+}\ket{1},\ket{-}\ket{0},\ket{-}\ket{1}\}$.
We can use this information, together with the robust phase estimation algorithm in \cite{KimmelLowYoder2015robust}, to estimate the differences between pairs of eigenvalues, which in turn yield the parameters $\lambda_{XZ},\lambda_{XI},\lambda_{IZ}$ through a Hadamard transform.
The procedure for applying random Pauli operators and obtaining parameters from eigenvalue estimation are described in detail in Sections \ref{sec:isolate_diagonal} and \ref{sec:the_diagonal} respectively.
By using different choices of $U_1, \ldots, U_4$ to reshape $H$,
we can get all the parameters $\lambda_{P P'}$ in the two-qubit Hamiltonian $H$.
The same idea generalizes to arbitrary Hamiltonians on $\mathcal{O}(1)$ qubits.

\subsection{Learning a many-qubit Hamiltonian through divide and conquer}

If we want to learn a Hamiltonian on many qubits by directly applying the above method, the total evolution time will scale exponentially with the number of qubits.
Here, we present a divide-and-conquer approach to solving this problem.
To illustrate the proposed approach, let us consider a simple example of an inhomogeneous Heisenberg model on $N$ qubit with a Hamiltonian given by,
\begin{equation} \label{eq:Heisenberg-H}
    H = \sum_{\alpha=1}^{N-1}\left(\lambda^{\alpha,\alpha+1}_{x} X_{\alpha}X_{\alpha+1} + \lambda^{\alpha,\alpha+1}_{y} Y_{\alpha}Y_{\alpha+1} + \lambda^{\alpha,\alpha+1}_{z} Z_{\alpha}Z_{\alpha+1}\right) + \sum_{\alpha=1}^N \lambda^{\alpha}_z Z_{\alpha}
\end{equation}
where $\lambda^{\alpha,\alpha+1}_{x}, \lambda^{\alpha,\alpha+1}_{y}, \lambda^{\alpha,\alpha+1}_{z}, \lambda^{\alpha}_{z}$ are the unknown parameters.
Suppose we want to learn the parameter $\lambda^{12}_{x}$ on the first two qubits.
In order to achieve this, we reshape the unknown Hamiltonian $H$ with $U_1 = I, U_2 = X_3, U_3 = Y_3, U_4 = Z_3$ and $w_1 = w_2 = w_3 = w_4 = \tfrac{1}{4}$.
The new unknown Hamiltonian after the reshaping is given by
\begin{align}
    \widetilde{H} &= \frac{1}{4}(H+X_3 HX_3 + Y_3 H Y_3 + Z_3 H Z_3) \\
    &= \left(\lambda^{1,2}_{x} X_{1}X_{2} + \lambda^{1,2}_{x} Y_{1}Y_{2} + \lambda^{1,2}_{x} Z_{1}Z_{2}\right) + \lambda^{1}_z Z_{1} + \lambda^{2}_z Z_{2} \\
    &+ \sum_{\alpha=4}^{N-1}\left(\lambda^{\alpha,\alpha+1}_{x} X_{\alpha}X_{\alpha+1} + \lambda^{\alpha,\alpha+1}_{y} Y_{\alpha}Y_{\alpha+1} + \lambda^{\alpha,\alpha+1}_{z} Z_{\alpha}Z_{\alpha+1}\right) + \sum_{\alpha=4}^N \lambda^{\alpha}_z Z_{\alpha}.
\end{align}
The second equality above can be seen as follows:
For each Pauli operator $P \in \{I, X, Y, Z\}^{\otimes N}$, if it acts non-trivially on the third qubit, then we can show that
\begin{equation}
    \frac{1}{4}(P+X_3 PX_3 + Y_3 P Y_3 + Z_3 P Z_3) = 0.
\end{equation}
On the other hand, for Pauli operator $P$ that acts as identity on the third qubit, we can show that
\begin{equation}
    \frac{1}{4}(P+X_3 PX_3 + Y_3 P Y_3 + Z_3 P Z_3) = P.
\end{equation}
Now, we can see that the new unknown Hamiltonian $\widetilde{H}$ can be written as
\begin{equation}
    \widetilde{H} = \widetilde{H}_{12} + \widetilde{H}_{\geq 4},
\end{equation}
where $\widetilde{H}_{12}$ is an $N$-qubit Hamiltonian acting only on qubit $1$ and $2$, and $\widetilde{H}_{\geq 4}$ is an $N$-qubit Hamiltonian acting only on qubit $4, 5, \ldots$.
Therefore in the new Hamiltonian $\widetilde{H}$ after the reshaping, there is no entanglement between qubit $1$ and $2$ with the rest of the system.
This enables us to apply the learning algorithm for few-qubit Hamiltonians to estimate $\lambda^{1,2}_{x}$.

We can apply the above idea to learn every parameter in the Hamiltonian with a number of experiments that scales linearly in the system size $N$ rather than exponential in $N$.
We show that one could do better than linear scaling with a parallelization technique.
In particular, we discuss how one could learn all the parameters $\lambda^{1,2}_x,\lambda^{4,5}_x,\lambda^{7,8}_x,\cdots$ in parallel.
Consider reshaping the unknown $N$-qubit Hamiltonian $H$ given in Eq.~\eqref{eq:Heisenberg-H} using
\begin{equation}
    U_1 = I, U_2 = X_3 X_6 X_9 \ldots, U_3 = Y_3 Y_6 Y_9 \ldots, U_4 = Z_3 Z_6 Z_9 \ldots,
\end{equation}
and $w_1 = w_2 = w_3 = w_4 = 1/4$.
Then the new Hamiltonian under reshaping is given by
\begin{equation}
    \widetilde{H} = \widetilde{H}_{12} + \widetilde{H}_{45} + \widetilde{H}_{78} + \ldots,
\end{equation}
where $\widetilde{H}_{\alpha, \alpha+1}$ is an $N$-qubit Hamiltonian of the form
\begin{equation} \label{eq:widetildeH-Heisenberg}
    \widetilde{H}_{\alpha, \alpha+1} = \left(\lambda^{\alpha,\alpha+1}_{x} X_{\alpha}X_{\alpha+1} + \lambda^{\alpha,\alpha+1}_{x} Y_{\alpha}Y_{\alpha+1} + \lambda^{\alpha,\alpha+1}_{x} Z_{\alpha}Z_{\alpha+1}\right) + \lambda^{\alpha}_z Z_{\alpha} + \lambda^{\alpha+1}_z Z_{\alpha+1}
\end{equation}
for all $\alpha = 1, 4, 7, \ldots$.
Using a reshaping based on four unitaries $U_1, \ldots, U_4$, we have turn the unknown $N$-qubit interacting Hamiltonian $H$ into a new Hamiltonian $\widetilde{H}$ with many noninteracting patches of two qubits.
Each two-qubit patch is now evolving independently from each other. This decoupling enables us to estimate the parameters in parallel.

This divide-and-conquer method works for any Hamiltonian that can be written as a sum of few-body observables, where each qubit is acted by at most $\mathcal{O}(1)$ of the few-body observables.
For this more general class of Hamiltonians, we determine how the reshaping is done by performing a coloring over its cluster interaction graph (Lemma \ref{lem:coloring_cluster_int_graph}).
For details, see Sections \ref{sec:low_intersection_ham} and \ref{sec:decoupling_dynamics_through_twirling}.
A complete description of our algorithm for the general low-intersection Hamiltonians can be found in Algorithm \ref{alg:ham_learn} in Section \ref{sec:est_all_bases_clusters}.
The cost of the algorithm is summarized in Theorem \ref{thm:ham_learn_partial_pauli_twirling}.

\subsection{Characterizing approximation error in reshaping Hamiltonians}
\label{sec:char-reshape}

The estimation error of the proposed learning algorithm depends on the quantum measurement error as well as the approximation error when we reshape the unknown Hamiltonian into other forms.
One way to analyze the approximation error is through the error analysis considered in \cite{Campbell2019random} if we use qDRIFT to reshape or in \cite{childs2019theory} when using the second-order Trotter formula.
However, these analyses are concerned with the error in the worst-case scenario over all possible input states and all observables.
For the learning task given here, it leads to an overestimation of the approximation error as some key properties of the problem are not incorporated.

Consider the example of learning an inhomogeneous Heisenberg model on $N$ qubit given in the previous section.
To evolve under the $N$-qubit Hamiltonian $\widetilde{H}$ in Eq.~\eqref{eq:widetildeH-Heisenberg} for time $t$,
the analysis in \cite{childs2019theory} shows that the approximation error of qDRIFT with $r$ steps is given by $\mathcal{O}(N^2 t^2 / r)$.
Here, $\widetilde{H}$ is decoupled into many two-qubit patches that do not interact with each other. And we are interested only in the accuracy in evolving each patch.
This prevents error from propagating across the entire $N$-qubit system.
A tighter analysis, using these facts, shows that the approximation error is given by $\mathcal{O}(t^2 / r)$ without an $N$ dependence.
We give the improved analysis for reshaping Hamiltonians using the randomization approach in Section~\ref{sec:partial_pauli_twirling} and Section~\ref{sec:deviate_limit_dynamics}.
The improved analysis for using the second-order Trotter formula is given in Section~\ref{sec:trotter} and Section~\ref{sec:deviate_limit_dynamics_trotter}.

\subsection{Establishing a matching lower bound}

We prove a matching lower bound of $T = \Omega(\epsilon^{-1} \log(\delta^{-1}) )$ on the total evolution time $T$.
The optimality with respect to the $\epsilon$ dependence is obtained by the Heisenberg limit.
However, the optimality concerning the failure probability $\delta$ has not been proven in the literature.
We consider any learning algorithm that can run new experiments based adaptively on the outcomes of previous experiments.
To handle adaptivity, we consider the rooted tree representation of the learning algorithm \cite{huang2022quantum, chen2022exponential}, and consider the task of distinguishing between two distinct Hamiltonians $H_\pm = \pm \epsilon Z$.

We begin by considering how well one could use a single experiment to distinguish $H_\pm$, which is characterized by the total variation distance between the probability distribution over experimental outcomes under $H_\pm$.
The single-experiment analysis establishes a relation between $\epsilon$ and the evolution time in one experiment.
We then consider an induction over every subtree of the learning algorithm to study the performance over multiple experiments.
A central technique is to control how each additional experiment improves one's ability to distinguish $H_\pm$.
The proof of the lower bound is given in Section~\ref{sec:lower_bound}.

\section{Outlook}

Our work shows that one can achieve the Heisenberg limit in learning many-body local Hamiltonian with many unknown parameters.
On the theoretical side, the central open question is whether and how one could achieve the Heisenberg limit for learning other classes of many-body Hamiltonians.
In an $N$-qubit Hamiltonian with all-to-all two-body interactions,
the proposed techniques allow one to achieve the Heisenberg limit at the expense of a quadratic dependence on system size $N$.
By using the reshaping approach to decouple every pair of qubits,
we can learn the two-body interactions with a total evolution time of $T = \mathcal{O}(N^2 \epsilon^{-1} \log(\delta^{-1}))$.
However, the following question remains open: Can we achieve a scaling of $T = \mathcal{O}(\epsilon^{-1} \log(\delta^{-1}))$ for learning $N$-qubit Hamiltonians with all-to-all interactions?
In addition to $N$-qubit Hamiltonians with all-to-all connections, can we learn fermionic or bosonic many-body Hamiltonians with Heisenberg-limited precision scaling?
Answering this question is essential for applications such as reconstructing the structure of large molecules or learning the interactions in an exotic quantum material.
Even more ambitiously, can one achieve a scaling of $T = \mathcal{O}(\epsilon^{-1} \log(\delta^{-1}))$ for learning the unknown parameters in an arbitrary $N$-qubit Hamiltonian without any structure?

There are several important future directions related to practical considerations.
Here, we assume experiments that can interleave unknown Hamiltonian evolution with controllable quantum circuits.
However, it is practically easier to implement the procedure if we only control the initial state and the final measurement basis.
This raises the question of whether one could achieve the scaling $T = \mathcal{O}(\epsilon^{-1} \log(\delta^{-1}))$ for learning $N$-qubit local Hamiltonian $H$ in a restricted model, where we choose the initial state, evolve under $U(t) = e^{-i t H}$ for a chosen $t$, and measure in a chosen basis.
While our learning algorithm interleaves unknown Hamiltonian evolution with single-qubit Clifford gates, implementing Clifford gates can be challenging in many analog quantum simulators, such as Rydberg atom systems \cite{Fendley2004,browaeys_many-body_2020,Schauss1455,Endres2016,bernien2017probing,Labuhn,Misha256,2020arXiv201212268S}.
Could we replace the single-qubit Clifford gates with other controllable unitary evolutions?
Understanding these questions will be crucial for physically achieving the Heisenberg limit in learning many-body Hamiltonians.

\subsection*{Acknowledgments:}
The authors thank Matthias Caro, Richard Kueng, Lin Lin, Jarrod McClean, Praneeth Netrapalli, and John Preskill for valuable input and inspiring discussions.
HH is supported by a Google Ph.D. fellowship.
YT is supported in part by the U.S. Department of Energy Office of Science (DE-SC0019374), Office of Advanced Scientific Computing Research (DE-SC0020290), Office of High Energy Physics (DE-ACO2-07CH11359), and under the Quantum System Accelerator project. DF is supported by NSF Quantum Leap Challenge Institute (QLCI) program under Grant No. OMA-2016245, NSF DMS-2208416, and a grant from the Simons Foundation under Award No. 825053.

\bibliographystyle{unsrt}
\bibliography{ref}

\begin{thebibliography}{10}

\bibitem{de2005quantum}
Mark de~Burgh and Stephen~D Bartlett.
\newblock Quantum methods for clock synchronization: Beating the standard
  quantum limit without entanglement.
\newblock {\em Physical Review A}, 72(4):042301, 2005.

\bibitem{valencia2004distant}
Alejandra Valencia, Giuliano Scarcelli, and Yanhua Shih.
\newblock Distant clock synchronization using entangled photon pairs.
\newblock {\em Applied Physics Letters}, 85(13):2655--2657, 2004.

\bibitem{leibfried2004toward}
Dietrich Leibfried, Murray~D Barrett, T~Schaetz, Joseph Britton, J~Chiaverini,
  Wayne~M Itano, John~D Jost, Christopher Langer, and David~J Wineland.
\newblock Toward heisenberg-limited spectroscopy with multiparticle entangled
  states.
\newblock {\em Science}, 304(5676):1476--1478, 2004.

\bibitem{bollinger1996optimal}
John~J Bollinger, Wayne~M Itano, David~J Wineland, and Daniel~J Heinzen.
\newblock Optimal frequency measurements with maximally correlated states.
\newblock {\em Physical Review A}, 54(6):R4649, 1996.

\bibitem{lee2002quantum}
Hwang Lee, Pieter Kok, and Jonathan~P Dowling.
\newblock A quantum {Rosetta} stone for interferometry.
\newblock {\em Journal of Modern Optics}, 49(14-15):2325--2338, 2002.

\bibitem{mckenzie2002experimental}
Kirk McKenzie, Daniel~A Shaddock, David~E McClelland, Ben~C Buchler, and
  Ping~Koy Lam.
\newblock Experimental demonstration of a squeezing-enhanced power-recycled
  michelson interferometer for gravitational wave detection.
\newblock {\em Physical review letters}, 88(23):231102, 2002.

\bibitem{holland1993interferometric}
MJ~Holland and K~Burnett.
\newblock Interferometric detection of optical phase shifts at the heisenberg
  limit.
\newblock {\em Physical review letters}, 71(9):1355, 1993.

\bibitem{wineland1992spin}
David~J Wineland, John~J Bollinger, Wayne~M Itano, FL~Moore, and Daniel~J
  Heinzen.
\newblock Spin squeezing and reduced quantum noise in spectroscopy.
\newblock {\em Physical Review A}, 46(11):R6797, 1992.

\bibitem{caves1981quantum}
Carlton~M Caves.
\newblock Quantum-mechanical noise in an interferometer.
\newblock {\em Physical Review D}, 23(8):1693, 1981.

\bibitem{boulant2003}
Nicolas Boulant, Timothy~F. Havel, Marco~A. Pravia, and David~G. Cory.
\newblock Robust method for estimating the {Lindblad} operators of a
  dissipative quantum process from measurements of the density operator at
  multiple time points.
\newblock {\em Physical Review A}, 67(4), April 2003.

\bibitem{innocenti2020}
Luca Innocenti, Leonardo Banchi, Alessandro Ferraro, Sougato Bose, and Mauro
  Paternostro.
\newblock Supervised learning of time-independent {Hamiltonians} for gate
  design.
\newblock {\em New Journal of Physics}, 22(6), June 2020.

\bibitem{ben2020}
Eitan Ben~Av, Yotam Shapira, Nitzan Akerman, and Roee Ozeri.
\newblock Direct reconstruction of the quantum-master-equation dynamics of a
  trapped-ion qubit.
\newblock {\em Physical Review A}, 101(6), June 2020.

\bibitem{shulman2014}
Michael~D. Shulman, S.~P. Harvey, J.~M. Nichol, S.~D. Bartlett, A.~C. Doherty,
  V.~Umansky, and A.~Yacoby.
\newblock Suppressing qubit dephasing using real-time {Hamiltonian} estimation.
\newblock {\em Nature Communications}, 5(1), December 2014.

\bibitem{sheldon2016}
Sarah Sheldon, Easwar Magesan, Jerry~M. Chow, and Jay~M. Gambetta.
\newblock Procedure for systematically tuning up cross-talk in the
  cross-resonance gate.
\newblock {\em Physical Review A}, 93(6), June 2016.

\bibitem{sundaresan2020}
Neereja Sundaresan, Isaac Lauer, Emily Pritchett, Easwar Magesan, Petar
  Jurcevic, and Jay~M. Gambetta.
\newblock Reducing {Unitary} and {Spectator} {Errors} in {Cross} {Resonance}
  with {Optimized} {Rotary} {Echoes}.
\newblock {\em PRX Quantum}, 1(2), December 2020.

\bibitem{zhang2022scalable}
Xueyue Zhang, Eunjong Kim, Daniel~K Mark, Soonwon Choi, and Oskar Painter.
\newblock A scalable superconducting quantum simulator with long-range
  connectivity based on a photonic bandgap metamaterial.
\newblock {\em arXiv preprint arXiv:2206.12803}, 2022.

\bibitem{wiebe2014a}
Nathan Wiebe, Christopher Granade, Christopher Ferrie, and David Cory.
\newblock Quantum hamiltonian learning using imperfect quantum resources.
\newblock {\em Physical Review A}, 89(4), April 2014.

\bibitem{wiebe2014b}
Nathan Wiebe, Christopher Granade, Christopher Ferrie, and D.~G. Cory.
\newblock Hamiltonian learning and certification using quantum resources.
\newblock {\em Physical Review Letters}, 112(19), May 2014.

\bibitem{verdon2019}
Guillaume Verdon, Jacob Marks, Sasha Nanda, Stefan Leichenauer, and Jack
  Hidary.
\newblock Quantum hamiltonian-based models and the variational quantum
  thermalizer algorithm, 2019.

\bibitem{burgarth2017}
Daniel Burgarth and Ashok Ajoy.
\newblock Evolution-{Free} {Hamiltonian} {Parameter} {Estimation} through
  {Zeeman} {Markers}.
\newblock {\em Physical Review Letters}, 119(3), July 2017.

\bibitem{wang2017}
Jianwei Wang, Stefano Paesani, Raffaele Santagati, Sebastian Knauer, Antonio~A.
  Gentile, Nathan Wiebe, Maurangelo Petruzzella, Jeremy~L. O'Brien, John~G.
  Rarity, Anthony Laing, and et~al.
\newblock Experimental quantum hamiltonian learning.
\newblock {\em Nature Physics}, 13(6), March 2017.

\bibitem{kwon2020}
Hee~Young Kwon, H.~G. Yoon, C.~Lee, G.~Chen, K.~Liu, A.~K. Schmid, Y.~Z. Wu,
  J.~W. Choi, and C.~Won.
\newblock Magnetic {Hamiltonian} parameter estimation using deep learning
  techniques.
\newblock {\em Science Advances}, 6(39), September 2020.

\bibitem{wang2020}
Dingchen Wang, Songrui Wei, Anran Yuan, Fanghua Tian, Kaiyan Cao, Qizhong Zhao,
  Yin Zhang, Chao Zhou, Xiaoping Song, Dezhen Xue, and Sen Yang.
\newblock Machine {Learning} {Magnetic} {Parameters} from {Spin}
  {Configurations}.
\newblock {\em Advanced Science}, 7(16), August 2020.

\bibitem{cotler2021emergent}
Jordan~S Cotler, Daniel~K Mark, Hsin-Yuan Huang, Felipe Hernandez, Joonhee
  Choi, Adam~L Shaw, Manuel Endres, and Soonwon Choi.
\newblock Emergent quantum state designs from individual many-body
  wavefunctions.
\newblock {\em arXiv preprint arXiv:2103.03536}, 2021.

\bibitem{choi2021emergent}
Joonhee Choi, Adam~L Shaw, Ivaylo~S Madjarov, Xin Xie, Jacob~P Covey, Jordan~S
  Cotler, Daniel~K Mark, Hsin-Yuan Huang, Anant Kale, Hannes Pichler, et~al.
\newblock Emergent randomness and benchmarking from many-body quantum chaos.
\newblock {\em arXiv preprint arXiv:2103.03535}, 2021.

\bibitem{huang2020predicting}
Hsin-Yuan Huang, Richard Kueng, and John Preskill.
\newblock Predicting many properties of a quantum system from very few
  measurements.
\newblock {\em Nat. Phys.}, 16:1050––1057, 2020.

\bibitem{HaahKothariTang2021optimal}
Jeongwan Haah, Robin Kothari, and Ewin Tang.
\newblock Optimal learning of quantum hamiltonians from high-temperature gibbs
  states.
\newblock {\em arXiv preprint arXiv:2108.04842}, 2021.

\bibitem{yu2022}
Wenjun Yu, Jinzhao Sun, Zeyao Han, and Xiao Yuan.
\newblock Practical and efficient hamiltonian learning, 2022.

\bibitem{hangleiter2021}
Dominik Hangleiter, Ingo Roth, Jens Eisert, and Pedram Roushan.
\newblock Precise hamiltonian identification of a superconducting quantum
  processor, 2021.

\bibitem{FrancaMarkovichEtAl2022efficient}
Daniel~Stilck Franca, Liubov~A Markovich, VV~Dobrovitski, Albert~H Werner, and
  Johannes Borregaard.
\newblock Efficient and robust estimation of many-qubit hamiltonians.
\newblock {\em arXiv preprint arXiv:2205.09567}, 2022.

\bibitem{ZubidaYitzhakiEtAl2021optimal}
Assaf Zubida, Elad Yitzhaki, Netanel~H Lindner, and Eyal Bairey.
\newblock Optimal short-time measurements for hamiltonian learning.
\newblock {\em arXiv preprint arXiv:2108.08824}, 2021.

\bibitem{BaireyAradEtAl2019learning}
Eyal Bairey, Itai Arad, and Netanel~H Lindner.
\newblock Learning a local hamiltonian from local measurements.
\newblock {\em Physical review letters}, 122(2):020504, 2019.

\bibitem{GranadeFerrieWiebeCory2012robust}
Christopher~E Granade, Christopher Ferrie, Nathan Wiebe, and David~G Cory.
\newblock Robust online hamiltonian learning.
\newblock {\em New Journal of Physics}, 14(10):103013, 2012.

\bibitem{gu2022practical}
Andi Gu, Lukasz Cincio, and Patrick~J Coles.
\newblock Practical black box hamiltonian learning.
\newblock {\em arXiv preprint arXiv:2206.15464}, 2022.

\bibitem{wilde2022learnH}
Frederik Wilde, Augustine Kshetrimayum, Ingo Roth, Dominik Hangleiter, Ryan
  Sweke, and Jens Eisert.
\newblock Scalably learning quantum many-body hamiltonians from dynamical data,
  2022.

\bibitem{KrastanovZhouEtAl2019stochastic}
Stefan Krastanov, Sisi Zhou, Steven~T Flammia, and Liang Jiang.
\newblock Stochastic estimation of dynamical variables.
\newblock {\em Quantum Science and Technology}, 4(3):035003, 2019.

\bibitem{giovannetti2011advances}
Vittorio Giovannetti, Seth Lloyd, and Lorenzo Maccone.
\newblock Advances in quantum metrology.
\newblock {\em Nature photonics}, 5(4):222--229, 2011.

\bibitem{Zhou2017AchievingTH}
Sisi Zhou, Mengzhen Zhang, John Preskill, and Liang Jiang.
\newblock Achieving the heisenberg limit in quantum metrology using quantum
  error correction.
\newblock {\em Nature Communications}, 9, 2017.

\bibitem{degen2017quantum}
Christian~L Degen, Friedemann Reinhard, and Paola Cappellaro.
\newblock Quantum sensing.
\newblock {\em Reviews of modern physics}, 89(3):035002, 2017.

\bibitem{higgins2007entanglement}
Brendon~L Higgins, Dominic~W Berry, Stephen~D Bartlett, Howard~M Wiseman, and
  Geoff~J Pryde.
\newblock Entanglement-free heisenberg-limited phase estimation.
\newblock {\em Nature}, 450(7168):393--396, 2007.

\bibitem{KimmelLowYoder2015robust}
Shelby Kimmel, Guang~Hao Low, and Theodore~J Yoder.
\newblock Robust calibration of a universal single-qubit gate set via robust
  phase estimation.
\newblock {\em Physical Review A}, 92(6):062315, 2015.

\bibitem{anshu2021improved}
Anurag Anshu, David Gosset, Karen J~Morenz Korol, and Mehdi Soleimanifar.
\newblock Improved approximation algorithms for bounded-degree local
  hamiltonians.
\newblock {\em Physical Review Letters}, 127(25):250502, 2021.

\bibitem{harrow2017extremal}
Aram~W Harrow and Ashley Montanaro.
\newblock Extremal eigenvalues of local hamiltonians.
\newblock {\em Quantum}, 1:6, 2017.

\bibitem{aharonov2022quantum}
Dorit Aharonov, Jordan Cotler, and Xiao-Liang Qi.
\newblock Quantum algorithmic measurement.
\newblock {\em Nature communications}, 13(1):1--9, 2022.

\bibitem{chen2022exponential}
Sitan Chen, Jordan Cotler, Hsin-Yuan Huang, and Jerry Li.
\newblock Exponential separations between learning with and without quantum
  memory.
\newblock In {\em 2021 IEEE 62nd Annual Symposium on Foundations of Computer
  Science (FOCS)}, pages 574--585. IEEE, 2022.

\bibitem{huang2022foundations}
Hsin-Yuan Huang, Steven~T Flammia, and John Preskill.
\newblock Foundations for learning from noisy quantum experiments.
\newblock {\em arXiv preprint arXiv:2204.13691}, 2022.

\bibitem{huang2022quantum}
Hsin-Yuan Huang, Michael Broughton, Jordan Cotler, Sitan Chen, Jerry Li, Masoud
  Mohseni, Hartmut Neven, Ryan Babbush, Richard Kueng, John Preskill, et~al.
\newblock Quantum advantage in learning from experiments.
\newblock {\em Science}, 376(6598):1182--1186, 2022.

\bibitem{Campbell2019random}
Earl Campbell.
\newblock Random compiler for fast hamiltonian simulation.
\newblock {\em Physical review letters}, 123(7):070503, 2019.

\bibitem{BerryChildsEtAl2020time}
Dominic~W Berry, Andrew~M Childs, Yuan Su, Xin Wang, and Nathan Wiebe.
\newblock Time-dependent {Hamiltonian} simulation with {$L^{1}$}-norm scaling.
\newblock {\em Quantum}, 4:254, 2020.

\bibitem{chen2021concentration}
Chi-Fang Chen, Hsin-Yuan Huang, Richard Kueng, and Joel~A Tropp.
\newblock Concentration for random product formulas.
\newblock {\em PRX Quantum}, 2(4):040305, 2021.

\bibitem{Suzuki91}
Masuo Suzuki.
\newblock General theory of fractal path integrals with applications to
  many‐body theories and statistical physics.
\newblock {\em Journal of Mathematical Physics}, 32(2):400--407, 1991.

\bibitem{Lloyd96}
Seth Lloyd.
\newblock Universal quantum simulators.
\newblock {\em Science}, 273(5278):1073--1078, 1996.

\bibitem{BACS05}
Dominic~W. Berry, Graeme Ahokas, Richard Cleve, and Barry~C. Sanders.
\newblock Efficient quantum algorithms for simulating sparse {H}amiltonians.
\newblock {\em Communications in Mathematical Physics}, 270(2):359--371, 2007.

\bibitem{childs2019theory}
Andrew~M. Childs, Yuan Su, Minh~C. Tran, Nathan Wiebe, and Shuchen Zhu.
\newblock Theory of {T}rotter error with commutator scaling.
\newblock {\em Phys. Rev. X}, 11:011020, Feb 2021.

\bibitem{Tran21}
Minh~C. Tran, Yuan Su, Daniel Carney, and Jacob~M. Taylor.
\newblock Faster digital quantum simulation by symmetry protection.
\newblock {\em PRX Quantum}, 2:010323, Feb 2021.

\bibitem{Burgarth2022oneboundtorulethem}
Daniel Burgarth, Paolo Facchi, Giovanni Gramegna, and Kazuya Yuasa.
\newblock One bound to rule them all: from {A}diabatic to {Z}eno.
\newblock {\em {Quantum}}, 6:737, June 2022.

\bibitem{mohseni2008quantum}
Masoud Mohseni, Ali~T Rezakhani, and Daniel~A Lidar.
\newblock Quantum-process tomography: Resource analysis of different
  strategies.
\newblock {\em Phys. Rev. A}, 77(3):032322, 2008.

\bibitem{Scott08}
A.~J. {Scott}.
\newblock {Optimizing quantum process tomography with unitary 2-designs}.
\newblock {\em J. Phys.}, A41:055308, 2008.

\bibitem{o2004quantum}
Jeremy~L O'Brien, Geoff~J Pryde, Alexei Gilchrist, Daniel~FV James, Nathan~K
  Langford, Timothy~C Ralph, and Andrew~G White.
\newblock Quantum process tomography of a controlled-not gate.
\newblock {\em Physical review letters}, 93(8):080502, 2004.

\bibitem{levy2021classical}
Ryan Levy, Di~Luo, and Bryan~K Clark.
\newblock Classical shadows for quantum process tomography on near-term quantum
  computers.
\newblock {\em arXiv preprint arXiv:2110.02965}, 2021.

\bibitem{merkel2013self}
Seth~T Merkel, Jay~M Gambetta, John~A Smolin, Stefano Poletto, Antonio~D
  C{\'o}rcoles, Blake~R Johnson, Colm~A Ryan, and Matthias Steffen.
\newblock Self-consistent quantum process tomography.
\newblock {\em Physical Review A}, 87(6):062119, 2013.

\bibitem{nielsen2020gate}
Erik Nielsen, John~King Gamble, Kenneth Rudinger, Travis Scholten, Kevin Young,
  and Robin Blume-Kohout.
\newblock Gate set tomography.
\newblock {\em arXiv preprint arXiv:2009.07301}, 2020.

\bibitem{blume2017demonstration}
Robin Blume-Kohout, John~King Gamble, Erik Nielsen, Kenneth Rudinger, Jonathan
  Mizrahi, Kevin Fortier, and Peter Maunz.
\newblock Demonstration of qubit operations below a rigorous fault tolerance
  threshold with gate set tomography.
\newblock {\em Nature communications}, 8(1):1--13, 2017.

\bibitem{Fendley2004}
Paul Fendley, K.~Sengupta, and Subir Sachdev.
\newblock Competing density-wave orders in a one-dimensional hard-boson model.
\newblock {\em Phys. Rev. B}, 69:075106, 2004.

\bibitem{browaeys_many-body_2020}
Antoine Browaeys and Thierry Lahaye.
\newblock Many-body physics with individually controlled {Rydberg} atoms.
\newblock {\em Nat. Phys.}, 16(2):132--142, 2020.

\bibitem{Schauss1455}
P.~Schau{\ss}, J.~Zeiher, T.~Fukuhara, S.~Hild, M.~Cheneau, T.~Macr{\`\i},
  T.~Pohl, I.~Bloch, and C.~Gross.
\newblock Crystallization in ising quantum magnets.
\newblock {\em Science}, 347(6229):1455--1458, 2015.

\bibitem{Endres2016}
Manuel Endres, Hannes Bernien, Alexander Keesling, Harry Levine, Eric~R
  Anschuetz, Alexandre Krajenbrink, Crystal Senko, Vladan Vuletic, Markus
  Greiner, and Mikhail~D Lukin.
\newblock {Atom-by-atom assembly of defect-free one-dimensional cold atom
  arrays}.
\newblock {\em Science}, 354(6315):1024--1027, 2016.

\bibitem{bernien2017probing}
Hannes Bernien, Sylvain Schwartz, Alexander Keesling, Harry Levine, Ahmed
  Omran, Hannes Pichler, Soonwon Choi, Alexander~S Zibrov, Manuel Endres,
  Markus Greiner, et~al.
\newblock Probing many-body dynamics on a 51-atom quantum simulator.
\newblock {\em Nature}, 551(7682):579--584, 2017.

\bibitem{Labuhn}
Henning Labuhn, Daniel Barredo, Sylvain Ravets, Sylvain de~L{\'e}s{\'e}leuc,
  Tommaso Macr{\`\i}, Thierry Lahaye, and Antoine Browaeys.
\newblock Tunable two-dimensional arrays of single rydberg atoms for realizing
  quantum ising models.
\newblock {\em Nature}, 534:667 EP --, 2016.

\bibitem{Misha256}
Sepehr {Ebadi}, Tout~T. {Wang}, Harry {Levine}, Alexander {Keesling}, Giulia
  {Semeghini}, Ahmed {Omran}, Dolev {Bluvstein}, Rhine {Samajdar}, Hannes
  {Pichler}, Wen~Wei {Ho}, Soonwon {Choi}, Subir {Sachdev}, Markus {Greiner},
  Vladan {Vuletic}, and Mikhail~D. {Lukin}.
\newblock {Quantum Phases of Matter on a 256-Atom Programmable Quantum
  Simulator}.
\newblock {\em arXiv e-prints}, page arXiv:2012.12281, 2020.

\bibitem{2020arXiv201212268S}
Pascal {Scholl}, Michael {Schuler}, Hannah~J. {Williams}, Alexander~A.
  {Eberharter}, Daniel {Barredo}, Kai-Niklas {Schymik}, Vincent {Lienhard},
  Louis-Paul {Henry}, Thomas~C. {Lang}, Thierry {Lahaye}, Andreas~M.
  {L{\"a}uchli}, and Antoine {Browaeys}.
\newblock {Programmable quantum simulation of 2D antiferromagnets with hundreds
  of Rydberg atoms}.
\newblock {\em arXiv e-prints}, page arXiv:2012.12268, 2020.

\bibitem{Leveque2007finite}
Randall~J LeVeque.
\newblock {\em Finite difference methods for ordinary and partial differential
  equations: steady-state and time-dependent problems}.
\newblock SIAM, 2007.

\bibitem{knill2008randomized}
Emanuel Knill, Dietrich Leibfried, Rolf Reichle, Joe Britton, R~Brad Blakestad,
  John~D Jost, Chris Langer, Roee Ozeri, Signe Seidelin, and David~J Wineland.
\newblock Randomized benchmarking of quantum gates.
\newblock {\em Physical Review A}, 77(1):012307, 2008.

\bibitem{magesan2011scalable}
Easwar Magesan, Jay~M Gambetta, and Joseph Emerson.
\newblock Scalable and robust randomized benchmarking of quantum processes.
\newblock {\em Physical review letters}, 106(18):180504, 2011.

\bibitem{erhard2019characterizing}
Alexander Erhard, Joel~J Wallman, Lukas Postler, Michael Meth, Roman Stricker,
  Esteban~A Martinez, Philipp Schindler, Thomas Monz, Joseph Emerson, and
  Rainer Blatt.
\newblock Characterizing large-scale quantum computers via cycle benchmarking.
\newblock {\em Nature communications}, 10(1):1--7, 2019.

\bibitem{harper2020efficient}
Robin Harper, Steven~T Flammia, and Joel~J Wallman.
\newblock Efficient learning of quantum noise.
\newblock {\em Nature Physics}, 16(12):1184--1188, 2020.

\bibitem{elben2022randomized}
Andreas Elben, Steven~T Flammia, Hsin-Yuan Huang, Richard Kueng, John Preskill,
  Beno{\^\i}t Vermersch, and Peter Zoller.
\newblock The randomized measurement toolbox.
\newblock {\em arXiv preprint arXiv:2203.11374}, 2022.

\bibitem{nechita2018almost}
Ion Nechita, Zbigniew Pucha{\l}a, {\L}ukasz Pawela, and Karol {\.Z}yczkowski.
\newblock Almost all quantum channels are equidistant.
\newblock {\em Journal of Mathematical Physics}, 59(5):052201, 2018.

\bibitem{watrous2018theory}
John Watrous.
\newblock {\em The theory of quantum information}.
\newblock Cambridge university press, 2018.

\end{thebibliography}

\appendix

\section{Preliminaries}

We begin with definitions used throughout the work as well as a basic lemma that follows immediately from the chromatic number of a graph.

\subsection{Notations}
\label{sec:notations}

Throughout this work, we will write $Q^P$ to denote the set of all mappings from $P$ to $Q$ for finite sets $P$ and $Q$. We also denote $[N]={1,2,\cdots, N}$. For the product of a sequence of operators $O_1,O_2,\cdots,O_L$, we write
\begin{equation}
    \prod^{\leftarrow}_{1\leq l\leq L} O_l = O_L\cdots O_2 O_1,\quad \prod^{\rightarrow}_{1\leq l\leq L} O_l = O_1 O_2\cdots O_L.
\end{equation}
We generally omit the arrows when taking a product of commuting operators.
We use $[A,B]=AB-BA$ to denote the commutator between $A$ and $B$, and we also write $\mathrm{ad}_A(B)=[A,B]$.
Throughout this work when we say that an operator is diagonal relative to a basis, what we mean is:
\begin{defn}[Diagonal operator]
\label{defn:diagonal_operator}
Let $B=\{\ket{v_l}\}$ be a basis of a Hilbert space. We say an operator $O$ is \textit{diagonal} relative to $B$ if $B$ is an eigenbasis of $O$.
\end{defn}

For a subsystem $A$ of the $N$-qubit system we consider, we use $\Tr_{A}$ to denote the partial trace after tracing out $A$. By extension, we use $\Tr_{[N]\setminus A}$ to denote the partial trace after tracing out all qubits not contained in $A$.

We consider $I$ to be the identity matrix, $X$ to be the Pauli-X matrix, $Y$ to be the Pauli-Y matrix, and $Z$ to be the Pauli-Z matrix.
We consider an $N$-qubit Pauli operator $P$ to be an element in the set of $N$-qubit observables $\{I, X, Y, Z\}^{\otimes N}$. We also use subscript to denote which qubit the Pauli operator acts on. For example, we use $X_{\alpha}$ to denote the Pauli-X operator acting on qubit $\alpha$, and $\gamma_{\alpha}$, $\gamma\in\{I,X,Y,Z\}$, to denote all Pauli operators acting on this qubit.

\subsection{Low-intersection Hamiltonians}
\label{sec:low_intersection_ham}

We adopt the problem setup from Ref.~\cite{HaahKothariTang2021optimal}. We consider a \textit{low-intersection Hamiltonian} following the definition below.
\begin{defn}[Low-intersection Hamiltonian]
\label{defn:low_intersection_ham}
A low-intersection Hamiltonian acting on $N$ qubits is a Hamiltonian $H$ that takes the following form:
\begin{equation}
\label{eq:Ham_gen_low_intersect}
    H=\sum_{a=1}^M\lambda_a E_a
\end{equation} 
where each $E_a$ is an $N$-qubit Pauli operator acting non-trivially on at most $k=\Or(1)$ qubits, and for each $a$, $E_a$ overlaps with $\mathfrak{d}=\Or(1)$ of $E_b$'s.
\end{defn}
Following Ref.~\cite{HaahKothariTang2021optimal}, we assume that $E_a$'s are known a priori and the goal is to estimate $\lambda_a$ for each $a$.
Also, as a consequence of $k,\mathfrak{d}=\Or(1)$, we have $M=\Or(N)$.
Below we introduce a set $\mathcal{V}$ to describe how the qubits interact with each other.
\begin{defn}[Interacting cluster]
\label{defn:interacting_cluster}
For each $a$, let $\operatorname{Supp}(E_a)$ be the support of $E_a$, i.e., the collection of qubits on which $E_a$ acts nontrivially. From the set $\{\operatorname{Supp}(E_a)\}$, we remove all $\operatorname{Supp}(E_a)$ such that $\operatorname{Supp}(E_a)\subset \operatorname{Supp}(E_b)$ for some $b\in[M]$. The remaining $\operatorname{Supp}(E_a)$'s form the set $\mathcal{V}$.  Each element of $\mathcal{V}$ we call an \textit{interacting cluster}.
\end{defn}
From the above construction it is clear that $|\mathcal{V}|\leq M$.
We then define the \textit{cluster interaction graph} as follows.

\begin{defn}[Cluster interaction graph]
\label{defn:cluster_int_graph}
The cluster interaction $\mathcal{G}=(\mathcal{V},\mathcal{E})$ has interacting clusters (from $\mathcal{V}$ in Definition \ref{defn:interacting_cluster}) as its vertices. The set of edges $\mathcal{E}$ is defined as follows: for each pair of interacting clusters $C$ and $C'$ ($C\neq C'$) in $\mathcal{V}$, $(C,C')\in \mathcal{E}$ if $C\cap C'\neq \emptyset$ or if there exists $C''\in\mathcal{V}$ such that $C\cap C''\neq \emptyset$ and $C'\cap C''\neq \emptyset$.
\end{defn}
From the definition of the low-intersection Hamiltonian, the degree of $\mathcal{G}$, $\operatorname{deg}(\mathcal{G})$, is upper bounded by a constant that is independent of the system size $N$. More precisely, $\operatorname{deg}(\mathcal{G})\leq \mathfrak{d}^2$ where $\mathfrak{d}$ is defined in Definition \ref{defn:low_intersection_ham}. 

For parallel estimation of different interacting clusters, we need to color the graph $\mathcal{G}$ so that adjacent vertices have different colors. The number of colors $\chi$ needed, which is the chromatic number of the graph, satisfies $\chi\leq \operatorname{deg}(\mathcal{G})+1=\Or(1)$. Therefore we have the following lemma
\begin{lem}[Coloring of the cluster interaction graph]
\label{lem:coloring_cluster_int_graph}
$\mathcal{V}$ can be divided into disjoint union 
\begin{equation}
\label{eq:decomposition_color}
    \mathcal{V}=\bigsqcup_{c=1}^{\chi}\mathcal{V}_c,
\end{equation} 
where no two adjacent vertices are in the same $\mathcal{V}_c$. In other words, for any $C$ and $C'$ in $\mathcal{V}_c$, $C\cap C'=\emptyset$, and for any $C''\in\mathcal{V}$, either $C\cap C''=\emptyset$ or $C'\cap C''=\emptyset$. Moreover $\chi=\Or(1)$.
\end{lem}

\section{Reshaping Hamiltonians using randomization}
\label{sec:partial_pauli_twirling}

Below we describe how to reshape the unknown $N$-qubit Hamiltonian $H$ into a new Hamiltonian with a simpler form based on a randomized Hamiltonian simulation algorithm known as qDRIFT \cite{Campbell2019random}.
Given a probability distribution $\mathcal{D}$ over $N$-qubit Pauli operators $\{I, X, Y, Z\}^{\otimes N}$, we consider the new Hamiltonian after reshaping to be
\begin{equation}
    \widetilde{H}(\mathcal{D}) \triangleq \Exp_{P \sim \mathcal{D}} [P H P].
\end{equation}
The qDRIFT algorithm can approximate (as a quantum channel) dynamics under $\widetilde{H}(\mathcal{D})$ by dynamics under $H$ as follows,
\begin{equation}
    e^{-i t \widetilde{H}(\mathcal{D})} \approx e^{-i \tau P_{k_r} H P_{k_r} } \ldots e^{-i \tau P_{k_1} H P_{k_1} } = P_{k_r} e^{-i \tau H} P_{k_r} \ldots P_{k_1} e^{-i \tau H} P_{k_1},
\end{equation}
where $r$ is an integer that determines the approximation error (larger $r$ implies smaller error), $\tau \triangleq t / r$, and $P_{k_1}, \ldots, P_{k_r}$ are independent random Pauli operators sampled from $\mathcal{D}$.
In the original paper \cite{Campbell2019random} on qDRIFT, it was shown that the approximation holds when one considers the expectation of the unitary (treated as a quantum channel) over the random Pauli operators $P_{k_{1}}, \ldots, P_{k_{r}}$.
In a subsequent work \cite{chen2021concentration}, it was shown that the approximation holds even with a single realization of $P_{k_{1}}, \ldots, P_{k_{r}}$ with high probability.

By choosing different distribution $\mathcal{D}$, we can reshape the unknown Hamiltonian $H$ into new Hamiltonians with a much simplified form.
In particular, the reshaping technique is useful for:
(1) decoupling the $N$-qubit system into many few-qubit noninteracting patches, and (2) isolating the diagonal Hamiltonian in each of the few-qubit patches.

\subsection{Decoupling into noninteracting patches}
\label{sec:decoupling_dynamics_through_twirling}

Recall that for each color $c \in [\chi]$, $\mathcal{V}_c$ is a set of interacting clusters (i.e., few-qubit patches).
For each color $c \in [\chi]$, we define a distribution $\mathcal{D}_c$ over $P \in \{I, X, Y, Z\}^{\otimes N}$ as follows. For each qubit $\alpha \in [N]$,
\begin{itemize}
    \item If qubit $\alpha$ is in one of the interacting clusters in $\mathcal{V}_c$, we consider $P_\alpha = I$.
    \item If qubit $\alpha$ is not in any of the interacting clusters in $\mathcal{V}_c$, we sample $P_\alpha \in \{I, X, Y, Z\}$ uniformly.
\end{itemize}
Then we let $P=\prod_{\alpha}P_{\alpha}$.
We establish the following lemma.

\begin{lem}[Decoupling into noninteracting patches]
\label{lem:effective_ham_decouple}
Defining $\mathcal{D}_c$ as above,
we have
\begin{equation}
    \widetilde{H}(\mathcal{D}_c) = \Exp_{P \sim \mathcal{D}_c} [P H P] = \sum_{C\in\mathcal{V}_c} H_C,
\end{equation}
where $H_C \triangleq \sum_{a:\operatorname{Supp}(E_a)\subset C} \lambda_a E_a$ is the sum of all terms in $H$ that are supported on $C$.
\end{lem}
\begin{proof}
Recall that $H = \sum_{a} \lambda_a E_a$. For each $a$, we consider the following.
\begin{itemize}
    \item If $E_a$ acts non-trivially on a qubit that is not in any of the interacting clusters in $\mathcal{V}_c$, then there is $1/2$ probability that $P \sim \mathcal{D}_c$ commutes with $E_a$, so that $P E_a P=E_a$, and $1/2$ probability that $P \sim \mathcal{D}_c$ anti-commutes with $E_a$, so that $PE_a P=-E_a$.
    Consequently,
    \begin{equation}
        \Exp_{P \sim \mathcal{D}_c}[ P E_a P ] = \frac{1}{2}(E_a-E_a) = 0.
    \end{equation}
    \item  
    If $E_a$ acts trivially on all qubits that are not in any of the interacting clusters in $\mathcal{V}_c$, 
    then $P \sim \mathcal{D}_c$ always commutes with $E_a$ because the supports of these two operators do not overlap.
    As a result, we have
    \begin{equation}
        \Exp_{P \sim \mathcal{D}_c}[P E_a P] = E_a.
    \end{equation}
\end{itemize}
Therefore $\widetilde{H}(\mathcal{D}_c) = \Exp_{P \sim \mathcal{D}}[ PH P]$ contains only those terms that are supported on $\bigcup_{C\in\mathcal{V}_c}C$. 

Next we show that those terms are supported on only a single $C\in\mathcal{V}_c$. If $E_a$ is supported on both $C\in\mathcal{V}_c$ and $C'\in\mathcal{V}_c$, then the support of $E_a$ overlaps with both $C$ and $C'$, making them adjacent by Definition \ref{defn:cluster_int_graph}, which precludes them from being including in the same $\mathcal{V}_c$, thus resulting in contradiction. Therefore each $E_a$ is supported on only a single $C\in\mathcal{V}_c$.
\end{proof}

Recall that an interacting cluster $C \in \mathcal{V}_c$ is a set of at most $k$ qubits.
Hence $H_C$ is an $N$-qubit Hamiltonian that acts non-trivially on at most $k = \mathcal{O}(1)$ of qubits.
For each $c \in [\chi]$,
the evolution under the new Hamiltonian $\widetilde{H}(\mathcal{D}_c)$ after reshaping is given by
\begin{equation}
    e^{-i t \widetilde{H}(\mathcal{D}_c)} = \bigotimes_{C \in \mathcal{V}_c} e^{-i t H_C},
\end{equation}
which is decoupled into many few-qubit patches that do not interact with each other.
In our algorithm we will learn all $H_C$'s in parallel for a given $c\in[\chi]$.
Because we prepare product states in all experiments, and measure observables that are local to each $C\in\mathcal{V}_c$, we can perform all the experiments in parallel as long as we evolve for the same length of time $t$. 
To be more precise, in each experiment, we perform the evolution (in terms of the density operator)
\begin{equation}
    \rho(0)=\bigotimes_{C\in\mathcal{V}_c}\rho_C\mapsto e^{-i\widetilde{H} (\mathcal{D}_c)t}\rho(0)e^{i\widetilde{H}(\mathcal{D}_c) t}=\bigotimes_{C\in\mathcal{V}_c} e^{-iH_C t}\rho_Ce^{iH_C t},
\end{equation}
where $\rho(0)$ is the initial state, and $\rho_C$ is the initial state for each $C\in\mathcal{V}_c$. The qubits not contained in $\bigcup_{C\in\mathcal{V}_c}C$ are neglected because they are decoupled from the dynamics. We then measure observables $O_C$ (supported on $C$) for each $C$ individually. The quantities we extract from the experiments are
\begin{equation}
    \tr[O_C e^{-iH_C t}\rho_Ce^{iH_C t}] = \tr[(O_C\otimes I) e^{-i\widetilde{H}(\mathcal{D}_c) t}\rho(0)e^{-i\widetilde{H}(\mathcal{D}_c) t}],
\end{equation}
where the identity operator $I$ acts on $[N]\setminus C$.
We do not need to rerun the experiment for each $C$ because $O_C$'s commute with each other.
Therefore, from now on we focus on a single $C$ and discuss how to learn $H_C$.

\subsection{Isolating the diagonal Hamiltonian}
\label{sec:isolate_diagonal}

Recall from Section~\ref{sec:notations} that any $\gamma \in \{I, X, Y, Z\}^{C}$ is a function mapping from a subset of qubits $C \subseteq [N]$ to a Pauli operator $\{I, X, Y, Z\}$.
Using this notation, we can write down the Hamiltonian $H_C$ in the Pauli basis as follows,
\begin{equation}
    \label{eq:pauli_expansion_HC}
    H_C = \sum_{\gamma\in\{I, X, Y, Z\}^C} \lambda_{\gamma} \prod_{\alpha \in C} \gamma(\alpha)_{\alpha},
\end{equation}
where $\gamma(\alpha)_{\alpha}$ is the Pauli operator $\gamma(\alpha)$ acting on qubit $\alpha$.
Hence, learning $H_C$ is equivalent to learning $\lambda_{\gamma}$'s.
Each $\lambda_{\gamma}$ corresponds to $\lambda_a$ in \eqref{eq:Ham_gen_low_intersect} for some $a\in[M]$. More specifically, $\lambda_{\gamma}=\lambda_a$ for $a \in [M]$ with $E_a = \prod_{\alpha\in C} \gamma(\alpha)_{\alpha}$ (if there does not exist such an $a$ then $\lambda_{\gamma}=0$).

In order to learn $H_C$, we again utilize the reshaping technique.
We reshape the Hamiltonian $H_C$ into a easier-to-learn form using the following distributions.
Given $\gamma \in \{X, Y, Z\}^{C}$.
We define the distribution $\mathcal{D}_{C, \gamma}$ over $N$-qubit Pauli operator $Q$ as follows. For each qubit $i \in [N]$,
\begin{itemize}
    \item If qubit $\alpha$ is in $C$, we consider $Q_\alpha = I$ or $\gamma(\alpha)$ with equal probability.
    \item If qubit $\alpha$ is not in $C$, we consider $Q_\alpha = I$.
\end{itemize}
Then we let $Q=\prod_{\alpha}Q_{\alpha}$.
We can establish the following lemma showing the new Hamiltonian $\widetilde{H_C}(\mathcal{D}_{C, \gamma})$ after reshaping.
\begin{lem}[Isolating the diagonal Hamiltonian]
\label{lem:effective_ham_isolate_diagonal}
    Using the definition of $\mathcal{D}_{C, \gamma}$, we have
    \begin{equation}
        \label{eq:new_effective_ham_diag}
        \widetilde{H_C}(\mathcal{D}_{C, \gamma}) =
        \Exp_{Q \sim \mathcal{D}_{C, \gamma}}[Q H_C Q]
        = \sum_{b\in\{0,1\}^C} \lambda_{b} \prod_{\alpha\in C} (\gamma(\alpha)_{\alpha})^{b(\alpha)} \triangleq H^C_{\mathrm{diag}}(\gamma),
    \end{equation}
    where $\lambda_b = \lambda_{\gamma'}$ for some $\gamma'\in\{I, X, Y, Z\}^C$ given by
    \begin{equation}
    \gamma'(\alpha)=
    \begin{cases}
    I,& \text{if}\ b(\alpha)=0, \\
    \gamma(\alpha),& \text{if}\ b(\alpha)=1.
    \end{cases}
    \end{equation}
\end{lem}
\begin{proof}
 
Each Pauli operator in $H_C$ can be written as $P=\prod_{\alpha\in C}\gamma'(\alpha)_{\alpha}$ for some $\gamma' \in \{I, X, Y, Z\}^C$. If $\gamma(\alpha)\neq\gamma'(\alpha)$ and $\gamma'(\alpha)\neq I$ for any $\alpha\in C$, $P$ will commute with half of the $Q$'s and anti-commute with the other half (we can simply count for how many $\alpha$'s we have $\gamma(\alpha)\neq\gamma'(\alpha)$ and $\gamma'(\alpha)\neq I$; if the number is even, $P$ and $Q$ commute, and if the number is odd, they anti-commute). 

Therefore in $\Exp_{Q \sim \mathcal{D}_{C, \gamma}}[ Q H_C Q]$ all these terms cancel out, and only terms that are products of $\gamma(\alpha)_{\alpha}$ for $\alpha\in C$, i.e., the diagonal terms, remain. 
\end{proof}

Let us define the Pauli eigenbases of an interacting cluster $C$.
Using this definition, $B_C(\gamma)$ is the orthonormal eigenbasis for the diagonal Hamiltonian $H^C_{\mathrm{diag}}(\gamma)$.
This means we have acquired a very important knowledge about the new unknown Hamiltonian $H^C_{\mathrm{diag}}(\gamma)$ after reshaping $H_C$: we know the eigenbasis of $H^C_{\mathrm{diag}}(\gamma)$.
In constrast, we do not know what the eigenbases for $H_C$ are.

\begin{defn}[Pauli eigenbases of an interacting cluster]
\label{defn:Pauli_eigenbasis}
For any interacting cluster $C$, and $\gamma \in \{X,Y,Z\}^C$, we define $B_C(\gamma)$ to be the orthonormal basis that simultaneously diagonalizes $\gamma(\alpha)_\alpha, \forall \alpha\in C$. We denote the set of all such bases for $C$ by $\mathcal{B}_C=\{B_C(\gamma):\gamma\in\{X, Y, Z\}^C\}$.
\end{defn}

\subsection{Combining the two reshaping procedures}
\label{sec:summary_partial_pauli_twirl}

We can combine the two reshaping procedures into a single one.
Given a color $c \in [\chi]$, which corresponds to a set $\mathcal{V}_c$ of many interacting clusters $C$, and $\gamma_C \in \{X, Y, Z\}^C$ for every interacting cluster $C \in \mathcal{V}_c$.
We consider a distribution $\mathcal{D}_{c, \{\gamma_C\}_{C \in \mathcal{V}_c}}$ over $N$-qubit Pauli operators $\bar{P}$ defined by Algorithm~\ref{alg:generate_random_Pauli}.
In Algorithm \ref{alg:generate_random_Pauli}, Lines 2 to 3 are for generating the $P$ operator used in Section \ref{sec:decoupling_dynamics_through_twirling} to decouple each $C$ from the rest of the system, and Lines 4 to 6 are for generating the $Q$ operator used in Section \ref{sec:isolate_diagonal} to isolate the diagonal Hamiltonian.
The Pauli operator $\bar{P}$ generated from this algorithm is therefore a product of $P$ and $Q$.
Consequently, from Lemmas \ref{lem:effective_ham_decouple} and \ref{lem:effective_ham_isolate_diagonal}, we can establish the following lemma.

\begin{algorithm}[t]
\caption{Generating the random Pauli operators}
\label{alg:generate_random_Pauli}
\begin{algorithmic}[1]
\REQUIRE $c\in[\chi]$, $\gamma_C\in\{X,Y,Z\}^C$ for each $C\in\mathcal{V}_c$
\FOR{$\alpha\in[N]$}
\IF{$\alpha\notin \bigsqcup_{C\in\mathcal{V}_c}$}
\STATE Let $s_{\alpha}$ be uniformly randomly drawn from $\{I,X,Y,Z\}$;
\ELSE
\STATE Let $C$ be the interacting cluster containing $\alpha$;
\STATE Let $s_{\alpha}$ be uniformly randomly drawn from $\{I, \gamma_C(\alpha) \}$;
\ENDIF
\ENDFOR
\ENSURE $\bar{P}=\bigotimes_{\alpha\in[N]}s_{\alpha}$.
\end{algorithmic}
\end{algorithm}

\begin{lem}
    \label{lem:effective_ham_combined}
    Given the definition of $\mathcal{D}_{c, \{\gamma_C\}_{C \in \mathcal{V}_c}}$, we have
    \begin{equation}
        \widetilde{H}(\mathcal{D}_{c, \{\gamma_C\}_{C \in \mathcal{V}_c}}) = \Exp_{\bar{P} \sim \mathcal{D}_{C, \{\gamma_C\}_{C \in \mathcal{V}_c}}}\left[ \bar{P} H \bar{P}\right] = \sum_{C\in\mathcal{V}_c} H_{\mathrm{diag}}^C(\gamma_C),
    \end{equation}
    where each $H_{\mathrm{diag}}^C(\gamma_C)$ is defined in Eq.~\eqref{eq:new_effective_ham_diag}. 
\end{lem}
\begin{proof}
To show this lemma, first note that
\begin{equation}
\label{eq:expectation_composition}
    \Exp_{\bar{P} \sim \mathcal{D}_{C, \{\gamma_C\}_{C \in \mathcal{V}_c}}}\left[ \bar{P} H \bar{P}\right] = \Exp_{\substack{Q_C \sim \mathcal{D}_{C, \gamma_C},\\ \forall C \in \mathcal{V}_c}}\left[ \left(\prod_{C \in \mathcal{V}_c} Q_C\right) \Exp_{P \sim \mathcal{D}_{c}}[ P H P] \left(\prod_{C \in \mathcal{V}_c} Q_C\right) \right].
\end{equation}
All the Pauli operators $P$ and $Q_C$ for all $C \in \mathcal{V}_c$ have disjoint supports.
By Lemma \ref{lem:effective_ham_decouple} we have
\begin{equation}
\Exp_{P \sim \mathcal{D}_{c}}[ P H P] = \sum_{C\in\mathcal{V}_c}H_C.
\end{equation}
Then by Lemma~\ref{lem:effective_ham_isolate_diagonal} and Eq.~\eqref{eq:expectation_composition} we have
\begin{equation}
    \Exp_{\bar{P} \sim \mathcal{D}_{C, \{\gamma_C\}_{C \in \mathcal{V}_c}}}\left[ \bar{P} H \bar{P}\right] = \sum_{C \in \mathcal{V}_c} \Exp_{\substack{Q_C \sim \mathcal{D}_{C, \gamma_C}}}\left[ Q_C H_C Q_C \right]
    = \sum_{C\in\mathcal{V}_c} H_{\mathrm{diag}}^C(\gamma_C),
\end{equation}
which establishes the lemma.
\end{proof}

\section{Learning the unknown Hamiltonian after reshaping}
\label{sec:Learning_from_twirled_dynamics}

In this section we will discuss how to learn parameters (coefficients) from the new unknown Hamiltonians after reshaping.
We first present how the experiments are executed in Section \ref{sec:experiments}.  Then we give the procedure to estimate the coefficients of terms that are diagonal in a given Pauli eigenbasis for a single cluster in Section \ref{sec:the_diagonal}. Finally, we talk about how to estimate parameters for all clusters in parallel in Section \ref{sec:est_all_bases_clusters}.

\subsection{Executing quantum experiments}
\label{sec:experiments}

We begin by describing how the experiments are executed.
Given a color $c\in[\chi]$, and a Pauli assignment $\gamma_C\in\{X,Y,Z\}^C$ for each interacting cluster $C\in\mathcal{V}_c$, we initialize the quantum system in a product state which can be prepared using single-qubit Clifford gates.
We then evolve under the new unknown Hamiltonian $\widetilde{H}(\mathcal{D}_{c, \{\gamma_C\}_{C \in \mathcal{V}_c}})$ after reshaping for time $t$.
Based on the qDRIFT algorithm, we can approximate the unitary dynamics $e^{-i t \widetilde{H}(\mathcal{D}_{c, \{\gamma_C\}_{C \in \mathcal{V}_c}})}$ by
\begin{equation}
\bar{P_r}e^{-iH\tau}\bar{P_r} \ldots \bar{P_1}e^{-iH\tau}\bar{P_1} = e^{-i \tau \bar{P_r} H \bar{P_r}} \ldots e^{-i \tau \bar{P_1} H \bar{P_1}},
\end{equation}
where $r$ is a large integer, $\tau = t / r$, and $\bar{P_j}$ is a random $N$-qubit Pauli operator sampled from distribution $\mathcal{D}_{c, \{\gamma_{C}\}_{C \in \mathcal{V}_c}}$ using Algorithm \ref{alg:generate_random_Pauli}.
After evolving the system, we then measure an observable $O_C$, supported on $C$, for every $C\in\mathcal{V}_c$.
Because $O_C$'s do not overlap with each other, these measurements can be performed simulataneously.
In this way we are able to estimate parameters for all clusters in $\mathcal{V}_c$ in parallel.

Using a density matrix formulation, the experiment begins by preparing an initial state $\rho(0)=\bigotimes_{C\in\mathcal{V}_c}\rho_C\otimes\rho_{\mathrm{res}}$, where $\rho_{\mathrm{res}}$ is the state of the qubits not contained in $\bigsqcup_{C\in\mathcal{V}_c}C$.
After the randomized evolution, the final state before the measurements is given by
\begin{equation}
\label{eq:final_state_averaged}
    \rho(t) = \Exp_{\bar{P_j} \sim \mathcal{D}_{c, \{\gamma_{C}\}_{C \in \mathcal{V}_c}}} \left[\prod_{1\leq j\leq r}^{\leftarrow}e^{-i\bar{P}_j H\bar{P}_j\tau}\rho(0)\prod_{1\leq j\leq r}^{\rightarrow}e^{i\bar{P}_j H\bar{P}_j\tau}\right].
\end{equation}
In the limit of $\tau\to 0$ (equivalently $r\to\infty$), the system will evolve under $\widetilde{H}(\mathcal{D}_{c, \{\gamma_C\}_{C \in \mathcal{V}_c}}) = \sum_{C\in\mathcal{V}_c}H_{\mathrm{diag}}^C(\gamma_C)$ as shown in Lemma \ref{lem:effective_ham_combined}.
This means if we look at an observable supported on the interacting cluster $C\in\mathcal{V}_c$, its expectation value will be approximately $\Tr[O_C e^{-iH_{\mathrm{diag}}^C(\gamma_C) t}\rho_C e^{iH_{\mathrm{diag}}^C(\gamma_C) t}]$. Note that everything in this expression depends only on the cluster $C$.

In an actual experiment, $r$ cannot be infinite. The theorem below tells us how small $\tau$ (or equivalently how large $r=t/\tau$) needs to be for the above procedure to achieve a given accuracy $\varepsilon$.
The proof of this theorem is given in Section \ref{sec:deviate_limit_dynamics}.

\begin{thm}[Number of required random Pauli operators]
\label{thm:num_rand_pauli_needed}
There exists $r_0=\Or(t^2/\varepsilon)$ such that for any $r>r_0$, any initial state $\rho(0)$, any $C \in \mathcal{V}_c$, and any $O_C$ supported on $C$ with $\|O_C\|\leq 1$, we have
\begin{equation}
    \Big|\Tr[(O_C\otimes I)\rho(t)]-\Tr[O_C e^{-iH_{\mathrm{diag}}^C(\gamma_C) t}\rho_C e^{iH_{\mathrm{diag}}^C(\gamma_C) t}]\Big|\leq \varepsilon,
\end{equation}
where $\rho_C=\Tr_{[N]\setminus C}\rho(0)$.
\end{thm}

\subsection{Estimating the diagonal}
\label{sec:the_diagonal}

Let us focus on how to estimate the parameters for terms that are diagonal in a given Pauli eigenbasis $B_C(\gamma_C)$ as defined in Definition \ref{defn:Pauli_eigenbasis} for a cluster $C$. 
One advantage of having the system evolve under $ H^C_{\mathrm{diag}}$ is that we have access to its eigenstates, which we denote by $\ket{\xi}$
\begin{equation}
    \ket{\xi} = \bigotimes_{\alpha\in C}\ket{\psi_{\alpha}},
\end{equation}
for $\xi\in\{0,1\}^C$, where $\ket{\psi_{\alpha}}$ is the $(-1)^{\xi(\alpha)}$-eigenstate of $\gamma_C(\alpha)$. For example, if $\gamma_C(\alpha)=X$, then $\ket{\psi_{\alpha}}=\ket{+}$ if $\xi(\alpha)=0$, and $\ket{\psi_{\alpha}}=\ket{-}$ if $\xi(\alpha)=1$. Importantly, $\ket{\xi}$ can be prepared using a tensor product of single-qubit Clifford gates.
For each $\ket{\xi}$, the corresponding eigenvalue can be calculated through
\begin{equation}
     H^C_{\mathrm{diag}}(\gamma_C) \ket{\xi} = \sum_{b \in\{0,1\}^C} \lambda_{b} (-1)^{\xi\cdot b}\ket{\xi}, 
\end{equation}
where $\xi\cdot b=\sum_{\alpha\in C}\xi(\alpha)b(\alpha)$ is the inner product between $\xi$ and $b$. The eigenvalues are therefore
\begin{equation}
\label{eq:coef_to_eigval}
    \varepsilon_{\xi} =  \sum_{b\in\{0,1\}^C}  (-1)^{\xi\cdot b} \lambda_{b}.
\end{equation}
The eigenvalues and the parameters are therefore related via the Hadamard transform. We can recover the parameters from the eigenvalues through
\begin{equation}
\label{eq:eigval_to_coef}
    \lambda_{b} =  \frac{1}{2^{|C|}}\sum_{\xi\in\{0,1\}^C}  (-1)^{\xi\cdot b} \varepsilon_{\xi}.
\end{equation}

From the above discussion we can see that the parameters $\lambda_{b}$ can be estimated from the eigenvalues $\varepsilon_{\xi}$. Rather than estimating $\varepsilon_{\xi}$ directly, which is impossible because of the presence of a global gauge, we will estimate $\varepsilon_{\xi}-\varepsilon_{\xi'}$ for pairs of $\xi$ and $\xi'$. Moreover from \eqref{eq:eigval_to_coef} we can see that, with the exception of the global phase $\lambda_{0^C}$ (we denote by $0^C$ that maps all elements of $C$ to $0$), all other $\lambda_b$'s depend only on the differences between $\varepsilon_{\xi}$'s. To this end we need to prepare a superposition of $\ket{\xi}$ and $\ket{\xi'}$. We note that when the Hamming distance between $\xi$ and $\xi'$ is $1$, then this is easy to do, because $(\ket{\xi}+\ket{\xi'})/\sqrt{2}$ is still a product state, and each of its tensor product component can be prepared using a single Clifford gate.
We denote the unitary preparing this state by
\begin{equation}
\label{eq:defn_U_xixi'}
    U_{\xi\xi'}\ket{0^{|C|}} = \frac{1}{\sqrt{2}}(\ket{\xi}+\ket{\xi'}).
\end{equation}
This unitary is a tensor product of single-qubit Clifford gates.
Similarly we can construct a unitary in the form of single-qubit Clifford gates that satisfy
\begin{equation}
\label{eq:defn_V_xixi'}
    V_{\xi\xi'}\ket{0^{|C|}} = \frac{1}{\sqrt{2}}(\ket{\xi}+i\ket{\xi'}).
\end{equation}
This can be done by replacing the Hadamard gate in $U_{\xi\xi'}$ with $S\mathrm{H}$ where $S$ is the phase gate.

Now we run experiments as follows: starting from $\ket{0^{|C|}}$, we apply $U_{\xi\xi'}$, and then evolve with $e^{-i H^C_{\mathrm{diag}}(\gamma_C)t}$ (which is approximately obtained by randomly applying $P_j$ and $Q_j$ as discussed above). Then we apply $U_{\xi\xi'}^{\dagger}$, and measure all the $k$ qubits. The probability of all qubits being returned to the $0$ state is
\begin{equation}
\label{eq:cos_prob}
    |\braket{0^{|C|}|U_{\xi\xi'}^{\dagger}e^{-i H^C_{\mathrm{diag}}(\gamma_C)t}U_{\xi\xi'}|0^{|C|}}|^2=\frac{1}{2}(1+\cos((\varepsilon_{\xi}-\varepsilon_{\xi'})t)).
\end{equation}
Similarly we can design an experiment in which the probability of returning to $\ket{0^{|C|}}$ is
\begin{equation}
\label{eq:sin_prob}
    |\braket{0^{|C|}|V_{\xi\xi'}^{\dagger}e^{-i H^C_{\mathrm{diag}}(\gamma_C)t}U_{\xi\xi'}|0^{|C|}}|^2=\frac{1}{2}(1+\sin((\varepsilon_{\xi}-\varepsilon_{\xi'})t)).
\end{equation}
We let $\tau'=\pi/2^{|C|+2}$ so that $\tau'|\varepsilon_{\xi}-\varepsilon_{\xi'}|\leq \pi/2$ (we know from \eqref{eq:coef_to_eigval} that $|\varepsilon|\leq 2^{|C|}$). Then let $t=\ell\tau'$ for positive integer $\ell$, the two probabilities in \eqref{eq:cos_prob} and \eqref{eq:sin_prob} become
\begin{equation}
\label{eq:probabilities_cos_sin}
\begin{aligned}
    p_0(\ell) &= \frac{1}{2}(1+\cos(\ell\tau'(\varepsilon_{\xi}-\varepsilon_{\xi'}))), \\
    p_+(\ell) &= \frac{1}{2}(1+\sin(\ell\tau'(\varepsilon_{\xi}-\varepsilon_{\xi'}))),
\end{aligned}
\end{equation}
corresponding to the probabilities in \cite[Theorem I.1]{KimmelLowYoder2015robust}. Using the robust phase estimation technique in \cite[Theorem I.1]{KimmelLowYoder2015robust}, we can then estimate $\tau'(\varepsilon_{\xi}-\varepsilon_{\xi'})$ with standard deviation $\epsilon'\tau'/3$, by running $e^{-i H^C_{\mathrm{diag}}(\gamma_C)\tau'}$ $\Or(\epsilon'^{-1}\tau'^{-1})$ times. Therefore the total evolution time with $H$ is $\Or(\epsilon'^{-1}\tau'^{-1})\times \tau'=\Or(\epsilon'^{-1})$. The number of experiments scale like $\Or(\polylog(\epsilon'^{-1}\tau'^{-1}))=\Or(\poly(|C|+\log(\epsilon'^{-1})))\leq \Or(\poly(k+\log(\epsilon'^{-1})))$. Here we use the fact that $|C|\leq k$.

With this we can estimate $\varepsilon_{\xi}-\varepsilon_{\xi'}$ with standard deviation $\epsilon'/3$. Our ultimate goal is to ensure that the estimate has low error with high probability. Therefore we can repeat the experiment $\Or(\log(\vartheta^{-1}))$ times and take the median to ensure that the error is below $\epsilon'$ with probability at least $1-\vartheta$. 
In the procedure above, in order to estimate $\varepsilon_{\xi}-\varepsilon_{\xi'}$ to precision $\epsilon'$ with probability at least $1-\vartheta$, we need a total evolution time of
\begin{equation}
    \label{eq:total_evo_time_pair_diff}
    \Or(\epsilon'^{-1}\log(\vartheta^{-1})),
\end{equation}
and the number of experiments required is
\begin{equation}
    \label{eq:num_expr_pair_diff}
    \Or(\poly(k+\log(\epsilon'^{-1}))\log(\vartheta^{-1})).
\end{equation}

The above procedure only gets us the differences $\varepsilon_{\xi}-\varepsilon_{\xi'}$ for $\xi$ and $\xi'$ that differ by one bit. Next we will discuss how to estimate each $\varepsilon_{\xi}$.
Because the global phase is undetectable we can assume $\varepsilon_{0^{C}}=0$ (here $\varepsilon_{0^C}$ is the eigenvalue corresponding to the mapping that maps all elements of $C$ to $0$). We can then estimate each $\varepsilon_{\xi}$ by the Hamming weight of $\xi$. Starting with $w=1$, once we have $\varepsilon_{\xi'}$ for all $\xi'$ with Hamming weight $w-1$, we can estimate all $\varepsilon_{\xi}$ with Hamming weight $w$, by estimating $\varepsilon_{\xi}-\varepsilon_{\xi'}$ for some $\xi'$ that differs from $\xi$ by one bit and has Hamming weight $w-1$. This allows us to estimate all $\varepsilon_{\xi}$, each of which through
\begin{equation}
\varepsilon_{\xi} = \sum_{l=0}^{w-1}(\varepsilon_{\xi_{l+1}}-\varepsilon_{\xi_{l}}),
\end{equation}
where $\xi_w=\xi$, $\xi_0=0^{C}$ (which means $\xi_0$ maps all elements of $C$ to $0$), $\xi_l$ has Hamming weight $l$, and $\xi_{l+1}$ and $\xi_{l}$ differ by only one bit. Because the summand on the right-hand side has at most $|C|\leq k$ terms, we only need to estimate each $\varepsilon_{\xi_{l+1}}-\varepsilon_{\xi_{l}}$ to precision $\epsilon'=\epsilon/k$ to ensure that the final error is at most $\epsilon$. 

This procedure can be seen as traversing a shortest path tree: if we define a graph with all $\xi\in\{0,1\}^C$ as vertices, and link $\xi$ and $\xi'$ if their Hamming distance is $1$, we will have a $|C|$-hypercube. Then we can define the shortest path tree as follows.

\begin{defn}[Shortest path tree]
\label{defn:shortest_path_tree}
The shortest path tree $\mathcal{T}_{\mathrm{SPT}}^C=(\{0,1\}^C,\mathcal{E}_{SPT}^C)$ is a subgraph of the $|C|$-hypercube, with root $0^C$, and $\mathcal{E}_{SPT}^C$ is the set of edges. $\mathcal{T}_{\mathrm{SPT}}^C$ satisfies that the path from the root to each vertex in the tree has the shortest distance in the $|C|$-hypercube.
\end{defn}
For each $(\xi,\xi')\in \mathcal{E}_{SPT}^C$, we estimate $\varepsilon_{\xi}-\varepsilon_{\xi'}$, and with this we can obtain the value of any $\varepsilon_{\xi}$ by traversing the path leading from $0^C$ to $\xi$ in $\mathcal{T}_{\mathrm{SPT}}^C$.

There are in total $2^{|C|}-1$ pairs of $\xi$ and $\xi'$ such that we need to estimate $\varepsilon_{\xi}-\varepsilon_{\xi'}$, because a tree with $2^{|C|}$ nodes has $2^{|C|}-1$ edges. 
In order to ensure that each estimate of $\varepsilon_{\xi}$ has confidence level $1-\delta$, each $\varepsilon_{\xi}-\varepsilon_{\xi'}$ needs a confidence level of $1-\delta/k$ by union bound. 
Therefore, substituting $\epsilon'=\epsilon/k$ and $\vartheta=\delta/k$ into \eqref{eq:total_evo_time_pair_diff}, for each $\varepsilon_{\xi}-\varepsilon_{\xi'}$ the total evolution time we need is
\begin{equation}
\label{eq:evolution_time_single_cluster_basis}
    \Or(k\epsilon^{-1}\log(k\delta^{-1})),
\end{equation}
and the number of experiments needed is, by substituting into \eqref{eq:num_expr_pair_diff},
\begin{equation}
\label{eq:num_expr_single_cluster_basis}
    \Or(\poly(k+\log(k\epsilon^{-1}))\log(k\delta^{-1})).
\end{equation}
Once all $\varepsilon_{\xi}$ are estimated with precision $\epsilon$, we can get all $\lambda_{b}$ in \eqref{eq:new_effective_ham_diag} with precision $\epsilon$ through \eqref{eq:eigval_to_coef}. 

\subsection{Estimating for all bases and clusters}
\label{sec:est_all_bases_clusters}

In Section \ref{sec:the_diagonal} we have focused on a single interacting cluster and a fixed Pauli eigenbasis.
This procedure needs to be repeated for all interacting clusters $C$, the number of which is upper bounded by $M$, and for all $3^{|C|}$ possible choices of basis (there is a lot of double counting involved, for which further optimization may be possible), in order to cover the parameters of all terms involved in \eqref{eq:Ham_gen_low_intersect}. 
Note that in Sections \ref{sec:decoupling_dynamics_through_twirling} and \ref{sec:summary_partial_pauli_twirl} we have showed that interacting clusters within the same $\mathcal{V}_c$ (having the same color in the coloring of the cluster interaction graph $\mathcal{G}$) can be estimated in parallel. Therefore we only need an overhead of $\chi=\Or(\mathfrak{d}^2)$ (the chromatic number in Lemma \ref{lem:coloring_cluster_int_graph}) rather than $M$ to get all interacting clusters. 

\begin{algorithm}[ht]
\caption{Learning the Hamiltonian}
\label{alg:ham_learn}
\begin{algorithmic}[1]
\REQUIRE Low-intersection Hamiltonian $H$ (Definition \ref{defn:low_intersection_ham}).
\STATE Generate the cluster interaction graph $\mathcal{G}=(\mathcal{V},\mathcal{E})$ (Definition \ref{defn:cluster_int_graph});
\STATE Color the cluster interaction graph: $\mathcal{V}=\bigsqcup_{c\in[\chi]}\mathcal{V}_c$ (Lemma \ref{lem:coloring_cluster_int_graph});
\FOR{$C\in\mathcal{V}$}
    \STATE Generate $\mathcal{T}_{\mathrm{SPT}}^C=(\mathcal{V}_{\mathrm{SPT}}^C,\mathcal{E}_{\mathrm{SPT}}^C)$, the shortest path tree of the $|C|$-hypercube (Definition \ref{defn:shortest_path_tree});
\ENDFOR
\FOR{$c\in[\chi]$}
    \STATE Let $\mathcal{S}_C=\{(\gamma,\xi,\xi'):\gamma\in\{X,Y,Z\}^C,(\xi,\xi')\in\mathcal{E}_{\mathrm{SPT}}^C\}$ for each $C\in\mathcal{V}_c$;
    \WHILE{$\sum_{C\in\mathcal{V}_c}|\mathcal{S}_C|>0$}
        \FOR{$C\in\mathcal{V}_c$}
            \IF{$\mathcal{S}_C\neq\emptyset$}
                \STATE Choose $(\gamma_C,\xi_C,\xi_C')\in\mathcal{S}_C$;
                \STATE Discard $(\gamma_C,\xi_C,\xi_C')$ from $\mathcal{S}_C$;
            \ELSE
                \STATE Randomly draw $\gamma_C$ from $\{X,Y,Z\}^C$; \COMMENT{This step is merely for notation consistency; we can let $C$ remain idle when $\mathcal{S}_C=\emptyset$.}
            \ENDIF
        \ENDFOR
        \STATE Generate random Pauli operators $\{\bar{P}_j\}$ using Algorithm \ref{alg:generate_random_Pauli} (with $c$ and $\{\gamma_C\}$ as input);
        \STATE Use robust phase estimation \cite{KimmelLowYoder2015robust} to estimate $\varepsilon_{\xi_C}-\varepsilon_{\xi_C'}$ for all $C\in\mathcal{V}_c$ simultaneously, by letting the system evolve under $H$ and inserting the Pauli operators $\{\bar{P}_j\}$ (Section \ref{sec:the_diagonal}, insertion of random Pauli operators described in Section \ref{sec:summary_partial_pauli_twirl}); 
    \ENDWHILE
    \FOR{$C\in\mathcal{V}_c$, $\gamma_C\in\{X,Y,Z\}^C$}
        \STATE Use $\{\varepsilon_{\xi_C}-\varepsilon_{\xi_C'}: (\xi_C,\xi_C')\in\mathcal{E}_{\mathrm{SPT}}^C\}$ generated above to generate estimate $\hat{\lambda}_a$ for parameter $\lambda_a$ of each term supported on $C$ and diagonal in the Pauli eigenbasis $B_{C}(\gamma_C)$ (Sections \ref{sec:the_diagonal} and \ref{sec:est_all_bases_clusters}, for the Pauli eigenbasis see Definition \ref{defn:Pauli_eigenbasis});
    \ENDFOR
\ENDFOR
\ENSURE Estimate $\hat{\lambda}_a$ of $\lambda_a$ for each $a\in[M]$.
\end{algorithmic}
\end{algorithm}

We summarize our procedure in Algorithm~\ref{alg:ham_learn}. 
From \eqref{eq:evolution_time_single_cluster_basis} and \eqref{eq:num_expr_single_cluster_basis}, we can get the total evolution time and number of experiments needed to learn all the parameters to within error $\epsilon$, with a confidence level of $1-\delta$ for each estimate: they are respectively
\begin{equation}
    \label{eq:evolution_time_total}
    3^k\chi\times (2^k-1)\times \Or( k\epsilon^{-1}\log(k\delta^{-1})) = \Or(k 6^k\mathfrak{d}^2 \epsilon^{-1}\log(k\delta^{-1})),
\end{equation}
and
\begin{equation}
    \label{eq:num_expr_total}
    3^k\chi\times(2^k-1)\times \Or( \poly(k+\log(k\epsilon^{-1}))\log(k\delta^{-1}))=\Or(6^k\mathfrak{d}^2\poly(k+\log(k\epsilon^{-1}))\log(k\delta^{-1})).
\end{equation}
When $k=\Or(1)$ and $\mathfrak{d}=\Or(1)$, they become $\Or( \epsilon^{-1}(\log(\delta^{-1})))$ and $\Or( \polylog(\epsilon^{-1})\log(\delta^{-1}))$ respectively.

In the above analysis we only considered the $\tau\to 0$ limit, i.e., we apply random Pauli operators infinitely frequently. This is impossible in reality. Fortunately, the robust phase estimation algorithm we use is robust to error below a constant threshold. More precisely, in \eqref{eq:cos_prob} and \eqref{eq:sin_prob}, we can tolerate an error up to $1/\sqrt{8}$ \cite{KimmelLowYoder2015robust}. Therefore we only need to apply random Pauli operators with a finite frequency. Theorem \ref{thm:num_rand_pauli_needed} tells us what the necessary frequency is. Below we summarize the cost of our algorithm.

\begin{thm}
\label{thm:total_evolution_time_assuming_const_err}
    Assume the following:
    for any $t>0$, $c\in[\chi]$, $\{\gamma_C\}_{C\in\mathcal{V}_c}$ (with $\chi$ and $\mathcal{V}_c$ defined in Lemma \ref{lem:coloring_cluster_int_graph}), we can start from initial state $\rho(0)=\bigotimes_{C\in\mathcal{V}_c}\rho_C\otimes\rho_{\mathrm{res}}$ (each $\rho_C$ is a density matrix for $C$, and $\rho_{\mathrm{res}}$ is the density matrix for the qubits not contained in $\bigsqcup_{C\in\mathcal{V}_c}C$) and apply random Pauli operators so that at time $t$ the quantum system, evolving under Hamiltonian \eqref{eq:Ham_gen_low_intersect}, is in the state $\rho(t)$ satisfying
    \begin{equation}
    \label{eq:err_tolerance_in_thm}
        \Big|\Tr[O_C e^{-iH^C_{\mathrm{eff}}(\gamma_C)t}\rho_C e^{iH^C_{\mathrm{eff}}(\gamma_C)t}]-\Tr[(O_C\otimes I)\rho(t)]\Big|\leq \frac{1}{\sqrt{8}},
    \end{equation}
    where $H^C_{\mathrm{eff}}(\gamma_C)$ is 
    as defined in \eqref{eq:new_effective_ham_diag} (for $\gamma=\gamma_C$), and $O_C$ is any Hermitian operator supported on $C$ with $\|O_C\|\leq 1$. 
    Under this assumption,
    we can generate estimates $\{\hat{\lambda}_a\}$ for parameters $\{\lambda_a\}$ in \eqref{eq:Ham_gen_low_intersect}, such that
    \begin{equation}
    \Pr[|\hat{\lambda}_a-\lambda_a|>\epsilon]<\delta
    \end{equation}
    for all $a\in[M]$ with the following cost:
    \begin{enumerate}
        \item $\Or(k 6^k\mathfrak{d}^2 \epsilon^{-1}\log(k\delta^{-1}))$ total evolution time;
        \item $\Or(6^k\mathfrak{d}^2\poly(k+\log(k\epsilon^{-1}))\log(k\delta^{-1}))$ number of experiments.
    \end{enumerate}
\end{thm}

By Theorem \ref{thm:num_rand_pauli_needed}, the condition \eqref{eq:err_tolerance_in_thm} in the above theorem can be satisfied by choosing $r=\Or(t^2)$. Therefore we arrive at our main result:
\begin{thm}[Learning many-body Hamiltonian by reshaping with randomization]
    \label{thm:ham_learn_partial_pauli_twirling}
    Assume that $H$ is a low-intersection Hamiltonian defined in Definition \ref{defn:low_intersection_ham}. Then using Algorithm \ref{alg:ham_learn}, we can generate estimates $\{\hat{\lambda}_a\}$ for parameters $\{\lambda_a\}$ in \eqref{eq:Ham_gen_low_intersect}, such that
    \begin{equation}
    \Pr[|\hat{\lambda}_a-\lambda_a|>\epsilon]<\delta
    \end{equation}
    for all $a\in[M]$ with the following cost:
    \begin{enumerate}
        \item $\Or(\epsilon^{-1}\log(\delta^{-1}))$ total evolution time;
        \item $\Or(\polylog(\epsilon^{-1})\log(\delta^{-1}))$ number of experiments;
        \item $\Or(N\epsilon^{-2}\polylog(\epsilon^{-1})\log(\delta^{-1}))$ single-qubit Clifford gates;
    \end{enumerate}
    Moreover, this algorithm is robust against SPAM error.
\end{thm}

Note that in this theorem we assume $\mathfrak{d}=\Or(1)$ and $k=\Or(1)$, and therefore do not consider the dependence on these two parameters. 

\begin{proof}
The total evolution time and the number of experiments are direct consequences of Theorem \ref{thm:total_evolution_time_assuming_const_err}. Therefore we only need to focus on how many single-qubit Clifford gates are needed. For each experiment, we need $\Or(N)$ such gates in $U_{\xi\xi'}$ (defined in \eqref{eq:defn_U_xixi'}) to prepare the initial state and in $V_{\xi\xi'}$ (defined in \eqref{eq:defn_V_xixi'}) to perform measurements. These two tasks require $\Or(N\polylog(\epsilon^{-1})\log(\delta^{-1}))$ single-qubit Clifford gates as a result. For each experiment, if the time evolution goes from $0$ to $t$, then $r=\Or(t^2)$, meaning that we need $\Or(Nt^2)$ single-qubit Clifford gates to implement the random Pauli operators. $t= \Or(\epsilon^{-1})$ for all experiments due to \cite{KimmelLowYoder2015robust}, and therefore the total number of single-qubit Clifford gates is $\Or(N\epsilon^{-2})$ multiplied by the number of experiments $\Or(\polylog(\epsilon^{-1})\log(\delta^{-1}))$, yielding the scaling stated in the theorem.

To see why the algorithm is robust against SPAM error, note that the probabilities of the output distribution can differ from those in \eqref{eq:cos_prob} and \eqref{eq:sin_prob} by as much as $1/\sqrt{8}$, and the robust phase estimation algorithm in \cite{KimmelLowYoder2015robust} will still work. As a result our algorithm can tolerate SPAM error below the threshold $1/\sqrt{8}$.
\end{proof}

\section{Deviation from the limiting dynamics in the randomization approach}
\label{sec:deviate_limit_dynamics}

In this section we will prove Theorem \ref{thm:num_rand_pauli_needed}. In fact, we will prove a stronger result, as stated in the following theorem:
\begin{thm}
    \label{thm:err_bound_heisenberg_picture}
    We assume that $H$ is a low-intersection Hamiltonian as defined in Definition \ref{defn:low_intersection_ham}. We assume random Pauli operators $\bar{P}_l$, $1\leq l\leq r$, are generated independently and are identically distributed as $\bar{P}$, which satisfies
    \begin{equation}
    \label{eq:condition_decoupled_ham}
        \mathbb{E}[\bar{P}H\bar{P}] = H_{\mathrm{eff}}^C + H_{\mathrm{env}},
    \end{equation}
    where $H_{\mathrm{eff}}^C$ is supported on a subsystem $C$ ($|C|=\Or(1)$) and $H_{\mathrm{env}}$ is supported on the rest of the system. Then
    \begin{equation}
    \label{eq:heisenberg_picture_err_bound_final_time}
        \Bigg\|\mathbb{E}\left[\prod_{1\leq l\leq r}^{\rightarrow}(\bar{P}_l e^{iH\tau}\bar{P}_l)(O_C\otimes I)\prod_{1\leq l\leq r}^{\leftarrow}(\bar{P}_l e^{-iH\tau}\bar{P}_l)\right] - e^{iH^C_{\mathrm{eff}}t}O_C e^{-iH^C_{\mathrm{eff}}t}\otimes I \Bigg\| = \Or(t^2/r),
    \end{equation}
    for any $O_C$ supported on $C$ satisfying $\|O_C\|\leq 1$. In particular, the constant in $\Or(t^2/r)$ does not depend on the system size $N$ or the number of terms $M$.
\end{thm}

We will postpone proving this theorem to later.
As can be seen from \eqref{eq:heisenberg_picture_err_bound_final_time}, this theorem concerns the evolution of a local observable $O_C$ in the Heisenberg picture. 
At time $t$, with the system evolving under $H$ and random Pauli operators inserted, $O_C$ becomes
\begin{equation}
\mathbb{E}\left[\prod_{1\leq l\leq r}^{\rightarrow}(\bar{P}_l e^{iH\tau}\bar{P}_l)(O_C\otimes I)\prod_{1\leq l\leq r}^{\leftarrow}(\bar{P}_l e^{-iH\tau}\bar{P}_l)\right]
\end{equation}
in the Heisenberg picture
and in the $\tau\to 0$ limit it should converge to
\begin{equation}
e^{iH^C_{\mathrm{eff}}t}O_C e^{-iH^C_{\mathrm{eff}}t}\otimes I.
\end{equation}
What the above theorem says is the following: 
when then random Pauli operators $\bar{P}_l$'s are applied sufficiently frequently, the evolution of $O_C$ is entirely determined by the local effective Hamiltonian $H^C_{\mathrm{eff}}$ up to a small error. The local effective Hamiltonian $H^C_{\mathrm{eff}}$, in the context of our algorithm, is $H^C_{\mathrm{diag}}(\gamma)$ defined in \eqref{eq:new_effective_ham_diag}. If we turn our attention to the observable expectation value, then the above theorem directly enables us to bound the error in observable expectation value, through the following corollary:

\begin{cor}
    \label{cor:err_bound_expectation_val}
    Under the same assumption as Theorem \ref{thm:err_bound_heisenberg_picture}, if the system is initially in a state $\rho(0)$, and at time $t$ evolves to
    \begin{equation}
    \rho(t) = \mathbb{E}\left[\prod_{1\leq l\leq r}^{\leftarrow}e^{-i\bar{P}_j H\bar{P}_j\tau}\rho(0)\prod_{1\leq l\leq r}^{\rightarrow}e^{i\bar{P}_j H\bar{P}_j\tau}\right],
    \end{equation}
    then 
    \begin{equation}
    \label{eq:expectation_value_err_bound}
        \left|\Tr[(O_C\otimes I)\rho(t)]-\Tr[O_C e^{-iH_{\mathrm{eff}}^Ct}\rho_Ce^{iH_{\mathrm{eff}}^Ct}]\right|=\Or(t^2/r),
    \end{equation}
    where $\rho_C=\Tr_{\mathrm{env}}\rho(0)$ ($\Tr_{\mathrm{env}}$ denotes the partial trace after tracing out the system outside $C$), and $O_C$ is supported on $C$ with $\|O_C\|\leq 1$.
\end{cor}

Before we prove this corollary let us first introduce some notations. 
The actual dynamics of the operator $O_C$ supported on $C$ at time $t_u=u\tau$ for $1\leq u\leq r$, when the system is evolving under $H$ with random Pauli operators inserted as described in Section \ref{sec:summary_partial_pauli_twirl}, is described by
\begin{equation}
\label{eq:OC_dynamics_actual}
    O_C^{(u)}=\mathbb{E}\left[\prod_{1\leq l\leq u}^{\rightarrow}(\bar{P}_l e^{iH\tau}\bar{P}_l)(O_C\otimes I)\prod_{1\leq l\leq u}^{\leftarrow}(\bar{P}_l e^{-iH\tau}\bar{P}_l)\right],
\end{equation}
where $O_C^{(r)}$ is the operator we get at time $t$, i.e., the end of the experiment.
The limiting dynamics is, for $\tau\to 0$,
\begin{equation}
\label{eq:OC_dynamics_limit}
    O_C(t)=e^{iH^C_{\mathrm{eff}}t}O_C e^{-iH^C_{\mathrm{eff}}t}.
\end{equation}

\begin{proof}[Proof of Corollary \ref{cor:err_bound_expectation_val}]
By Theorem \ref{thm:err_bound_heisenberg_picture} we have $\|O_C^{(r)}-O_C(t)\otimes I\|=\Or(t^2/r)$. The left-hand side of \eqref{eq:expectation_value_err_bound} can be written as 
\begin{equation}
\begin{aligned}
&\Big|\Tr[\rho(0)O_C^{(r)}]-\Tr[\rho(0)(O_C(t)\otimes I)]\Big| \\
&= \Big|\Tr[\rho(0)(O_C^{(r)}-O_C(t)\otimes I)]\Big| \\
&\leq \|O_C^{(r)}-O_C(t)\otimes I\|.
\end{aligned}
\end{equation}
Therefore by Theorem \ref{thm:err_bound_heisenberg_picture} we arrive at \eqref{eq:expectation_value_err_bound}.
\end{proof}

 This corollary, in turn, directly implies Theorem \ref{thm:num_rand_pauli_needed}. 
 \begin{proof}[Proof of Theorem \ref{thm:num_rand_pauli_needed}]
 By Lemma \ref{lem:effective_ham_combined}, for a fixed cluster $C$, we can write
 \begin{equation}
 \mathbb{E}[\bar{P}H\bar{P}] = H_{\mathrm{diag}}^C(\gamma_C) + \sum_{C'\in\mathcal{V}_c,C'\neq C} H_{\mathrm{diag}}^{C'}(\gamma_{C'}).
 \end{equation}
 Here the first term on the right-hand side is supported only on $C$ and the support of the second term on the right-hand side does not overlap with $C$, by virtue of the coloring in Lemma \ref{lem:coloring_cluster_int_graph}.
 Therefore the effective Hamiltonian has the form as required in \eqref{eq:condition_decoupled_ham}. Thus by Corollary \ref{cor:err_bound_expectation_val} we have
 \begin{equation}
 \Big|\Tr[(O_C\otimes I)\rho(t)]-\Tr[O_C e^{-iH_{\mathrm{diag}}^C(\gamma_C) t}\rho_C e^{iH_{\mathrm{diag}}^C(\gamma_C) t}]\Big|=\Or(\frac{t^2}{r}).
 \end{equation}
 In order to ensure that $\Or(\frac{t^2}{r})\leq \varepsilon$, it suffices to choose $r\geq r_0$ for some $r_0=\Or(t^2/\varepsilon)$.
 \end{proof}
 
We will then set about to prove Theorem \ref{thm:err_bound_heisenberg_picture}. 
\begin{proof}[Proof of Theorem \ref{thm:err_bound_heisenberg_picture}]
Using the notation introduced in \eqref{eq:OC_dynamics_actual} and \eqref{eq:OC_dynamics_limit}, \eqref{eq:heisenberg_picture_err_bound_final_time} can be written as
\begin{equation}
\label{eq:actual_limit_err}
    \|O_C^{(r)}-O_C(t)\otimes I\|=\Or\left(\frac{t^2}{r}\right).
\end{equation}

We will prove this inequality in two steps. We define
\begin{equation}
\label{eq:bar_OC_defn}
    \bar{O}^{(u)}_C=\left(I+i\tau\mathrm{ad}_{H^C_{\mathrm{eff}}}\right)^u O_C,
\end{equation}
This operator can be seen as a result of simulating the dynamics of $O_C(t)$ up to first order using Euler's method. It satisfy the following recursion relation: 
\begin{equation}
\label{eq:bar_OC_recursion}
    \bar{O}^{(u)}_C = \bar{O}^{(u-1)}_C + i\tau [H^C_{\mathrm{eff}},\bar{O}^{(u-1)}_C]
\end{equation}
with $\bar{O}^{(0)}_C=O_C$.

In the first step we will show that
\begin{equation}
\label{eq:decoupling_err_bound}
    \|\bar{O}^{(r)}_C\otimes I - O^{(r)}_C\|=\Or\left(\frac{t^2}{r}\right).
\end{equation}
Note that the right-hand side does not depend on the system size. Because $\bar{O}^{(r)}_C\otimes I$ acts non-trivially only on the cluster $C$, what the above bound means is that $O^{(r)}_C$ approximately only acts non-trivially on $C$, despite the fact that the dynamics due to $H$ will spread $O_C$ to the rest of the system. This inequality will be proved as Lemma \ref{lem:decoupling_err_bound} in Section \ref{sec:decoupling_err}.

In the second step, we will show that
\begin{equation}
\label{eq:OCt_and_barOC_diff}
    \|O_C(t)-\bar{O}_C^{(r)}\|= \Or\left(\frac{t^2}{r}\right)
\end{equation}
Again the right-hand side does not depend on the system size. This inequality will be proved as Lemma \ref{lem:OCt_and_barOC_diff} in Section \ref{sec:local_dynamics_err}. For the above inequality, both $O_C(t)$ and $\bar{O}_C^{(r)}$ are local operators supported on $C$, and therefore it characterizes the deviation of the local dynamics from the limiting dynamics.
Combining \eqref{eq:decoupling_err_bound} and \eqref{eq:OCt_and_barOC_diff}, we have \eqref{eq:actual_limit_err} by the triangle inequality.
\end{proof}

\subsection{The decoupling error}
\label{sec:decoupling_err}

In this section we will prove \eqref{eq:decoupling_err_bound}. We restate it in the following lemma
\begin{lem}
\label{lem:decoupling_err_bound}
Under the same assumptions as in Theorem \ref{thm:err_bound_heisenberg_picture}, we have
\begin{equation*}
    \|\bar{O}^{(r)}_C\otimes I - O^{(r)}_C\|=\Or\left(\frac{t^2}{r}\right),
\end{equation*}
where $O^{(r)}_C$ and $\bar{O}^{(r)}_C$ are defined in \eqref{eq:OC_dynamics_actual} and \eqref{eq:bar_OC_defn} respectively.
\end{lem}

\begin{proof}
We define
\begin{equation}
    \label{eq:defn_Mu}
    M^{(u-1)} = \mathbb{E}[(\bar{P}_u e^{iH\tau}\bar{P}_u)(\bar{O}^{(u-1)}_C\otimes I)(\bar{P}_u e^{-iH\tau}\bar{P}_u)] - \bar{O}^{(u)}_C\otimes I,
\end{equation}
and
\begin{equation}
    \label{eq:defn_Ru}
    R^{(u)} = M^{(u-1)} + \mathbb{E}[(\bar{P}_u e^{iH\tau}\bar{P}_u)R^{(u-1)}(\bar{P}_u e^{-iH\tau}\bar{P}_u)],
\end{equation}
with $R^{(0)}=0$.
Then we can inductively verify that
\begin{equation}
\label{eq:OC_decomp}
    O^{(u)}_C = \bar{O}^{(u)}_C\otimes I + R^{(u)}.
\end{equation}
Therefore we only need to prove that $\|R^{(r)}\|=\Or\left(\frac{t^2}{r}\right)$.

We first bound $\|M^{(u-1)}\|$. Using Taylor expansion, we have
\begin{equation}
\label{eq:taylor_exp_single_step}
    \mathbb{E}[(\bar{P}_u e^{iH\tau}\bar{P}_u)(\bar{O}^{(u-1)}_C\otimes I)(\bar{P}_u e^{-iH\tau}\bar{P}_u)] 
    =\sum_{j=0}^{\infty}\frac{(i\tau)^j}{j!}\mathbb{E}[\bar{P}_u \mathrm{ad}_H^j(\bar{P}_u(\bar{O}^{(u-1)}_C\otimes I)\bar{P}_u)\bar{P}_u].
\end{equation}
From this we want to upper bound $\|\mathrm{ad}_H^j(\bar{P}_u(\bar{O}^{(u-1)}_C\otimes I)\bar{P}_u\|$ for $j\geq 2$. We have
\begin{equation}
    \label{eq:j_nested_comm}
    \mathrm{ad}_H^j(\bar{P}_u(\bar{O}^{(u-1)}_C\otimes I)\bar{P}_u) = \sum_{a_1,a_2,\cdots,a_j}\lambda_{a_j}\cdots \lambda_{a_1}[E_{a_j},\cdots[E_{a_1},\bar{P}_u(\bar{O}^{(u-1)}_C\otimes I)\bar{P}_u]\cdots].
\end{equation}
Note that for the right-hand side, most of the terms are zero. We need to figure out how many terms are non-zero. For $a_1$, we note that $\bar{P}_u(\bar{O}^{(u-1)}_C\otimes I)\bar{P}_u$ is supported on $C$, and therefore only terms that acts non-trivially with $C$ has non-zero contribution. Therefore we only need to consider $a_1$ such that $\operatorname{Supp} E_{a_1}\cap C\neq \emptyset$. For $a_2$, because $\cdots[E_{a_1},\bar{P}_u(\bar{O}^{(u-1)}_C\otimes I)\bar{P}_u]$ has support on $\operatorname{Supp} E_{a_1}\cup C$, we only need to consider $a_2$ such that $\operatorname{Supp} E_{a_2}\cap(\operatorname{Supp} E_{a_1}\cup C)\neq \emptyset$.
From this we can conclude that the only non-zero terms are for $\vec{a}=(a_1,a_2,\cdots,a_j)$, where $\operatorname{Supp} E_{a_v}\cap(\bigcup_{\nu<v}\operatorname{Supp} E_{a_{\nu}}\cup C)\neq \emptyset$. We denote by $\mathcal{A}_j$ the set of $\vec{a}$ satisfying the above condition, and from \eqref{eq:j_nested_comm} we have
\begin{equation*}
    \mathrm{ad}_H^j(\bar{P}_u(\bar{O}^{(u-1)}_C\otimes I)\bar{P}_u) = \sum_{\vec{a}\in\mathcal{A}_j}\lambda_{a_j}\cdots \lambda_{a_1}[E_{a_j},\cdots[E_{a_1},\bar{P}_u(\bar{O}^{(u-1)}_C\otimes I)\bar{P}_u]\cdots].
\end{equation*}
Note that $|\lambda_{a_j}\cdots \lambda_{a_1}|\leq 1$, and $[E_{a_j},\cdots[E_{a_1},\bar{P}_u(\bar{O}^{(u-1)}_C\otimes I)\bar{P}_u]\leq 2^j\|\bar{O}^{(u-1)}_C\|$. Therefore
\begin{equation}
\label{eq:upper_bound_j_nested_comm}
    \|\mathrm{ad}_H^j(\bar{P}_u(\bar{O}^{(u-1)}_C\otimes I)\bar{P}_u) \|\leq 2^j|\mathcal{A}_j|\|\bar{O}^{(u-1)}_C\|.
\end{equation}
We then count $|\mathcal{A}_j|$: for $a_1$, by Definition \ref{defn:low_intersection_ham}, there are at most $\mathfrak{d}+1$ choices because this is the number of terms that overlap with $C$ (which is the support of a certain term in $H$), and for $a_2$, there are at most $2(\mathfrak{d}+1)$ choices, because the second operator can either overlap with $C$ or the first operator. Going until $a_j$, we can see that we have at most $j!\mathfrak{d}^j$ choices. Consequently $|\mathcal{A}_j|\leq j!(\mathfrak{d}+1)^j$. Substituting this into \eqref{eq:upper_bound_j_nested_comm} and further into the remainders terms in \eqref{eq:taylor_exp_single_step}, we have
\begin{equation}
\label{eq:upper_bound_remainder_taylor}
    \sum_{j=2}^{\infty}\frac{\tau^j}{j!}\mathbb{E}[\|\bar{P}_u \mathrm{ad}_H^j(\bar{P}_u(\bar{O}^{(u-1)}_C\otimes I)\bar{P}_u)\bar{P}_u\|]\leq \sum_{j=2}^{\infty}(2(\mathfrak{d}+1)\tau)^j\|\bar{O}^{(u-1)}_C\| = \frac{(2(\mathfrak{d}+1)\tau)^2}{1-2(\mathfrak{d}+1)\tau }\|\bar{O}^{(u-1)}_C\|,
\end{equation}
For the first two terms in \eqref{eq:taylor_exp_single_step} corresponding to $j=0,1$, we compute what they are:
\begin{equation}
\begin{aligned}
&\sum_{j=0}^{1}\frac{(i\tau)^j}{j!}\mathbb{E}[\bar{P}_u \mathrm{ad}_H^j(\bar{P}_u(\bar{O}^{(u-1)}_C\otimes I)\bar{P}_u)\bar{P}_u] \\
&=\bar{O}^{(u-1)}_C\otimes I +i\tau [H^C_{\mathrm{eff}},\bar{O}^{(u-1)}_C]\otimes I = \bar{O}^{(u)}_C\otimes I,
\end{aligned}
\end{equation}
where in the first equality we have used \eqref{eq:condition_decoupled_ham}, and in the second equality \eqref{eq:bar_OC_defn}. Substituting this and \eqref{eq:upper_bound_remainder_taylor} into \eqref{eq:taylor_exp_single_step}, we have
\begin{equation}
    \label{eq:Mu_upper_bound_incomplete}
    \|M^{(u-1)}\| = \|\mathbb{E}[(\bar{P}_u e^{iH\tau}\bar{P}_u)(\bar{O}^{(u-1)}_C\otimes I)(\bar{P}_u e^{-iH\tau}\bar{P}_u)] - \bar{O}^{(u)}_C\otimes I\|\leq \frac{(2(\mathfrak{d}+1)\tau)^2}{1-2(\mathfrak{d}+1)\tau }\|\bar{O}^{(u-1)}_C\|.
\end{equation}

It still remains to bound $\|\bar{O}^{(u-1)}_C\|$. To simplify our discussion, we note that for sufficiently small $\mathfrak{d}\tau$ (smaller than a constant), $\frac{(2(\mathfrak{d}+1)\tau)^2}{1-2(\mathfrak{d}+1)\tau }\leq A_1\mathfrak{d}^2\tau^2$ for some constant $A_1$. From \eqref{eq:Mu_upper_bound_incomplete}, we have
\begin{equation}
\|\bar{O}^{(u)}_C\|\leq \|\mathbb{E}[(\bar{P}_u e^{iH\tau}\bar{P}_u)(\bar{O}^{(u-1)}_C\otimes I)(\bar{P}_u e^{-iH\tau}\bar{P}_u)]\| + A_1\mathfrak{d}^2\tau^2\|\bar{O}^{(u-1)}_C\|\leq (1+A_1\mathfrak{d}^2\tau^2)\|\bar{O}^{(u-1)}_C\|.
\end{equation}
Combining this with the assumption that $\|O_C\|\leq 1$, we have
\begin{equation} \label{eq:barOC_norm}
    \|\bar{O}^{(u)}_C\|\leq (1+A_1\mathfrak{d}^2\tau^2)^u\leq (1+A_1\mathfrak{d}^2t^2/r^2)^r\leq 2,
\end{equation}
for sufficiently small $\mathfrak{d}^2t^2/r$. Therefore
\begin{equation}
    \label{eq:Mu_upper_bound}
    \|M^{(u-1)}\| = \|\mathbb{E}[(\bar{P}_u e^{iH\tau}\bar{P}_u)(\bar{O}^{(u-1)}_C\otimes I)(\bar{P}_u e^{-iH\tau}\bar{P}_u)] - \bar{O}^{(u)}_C\otimes I\|\leq A_2\mathfrak{d}^2\tau^2,
\end{equation}
for some constant $A_2$.

With this we can now bound $R^{(u)}$. By \eqref{eq:defn_Ru}, we have
\begin{equation}
\|R^{(u)}\|\leq \|M^{(u-1)}\| + \|R^{(u-1)}\|.
\end{equation}
Therefore
\begin{equation}
    \label{eq:bound_Ru}
    \|R^{(u)}\| \leq \sum_{l=1}^{u-1}\|M^{(l)}\| \leq A_2u\mathfrak{d}^2\tau^2.
\end{equation}
In particular
\begin{equation}
    \label{eq:bound_Rr}
    \|R^{(r)}\| \leq  A_2\frac{\mathfrak{d}^2t^2}{r},
\end{equation}
which proves the lemma.
\end{proof}

\subsection{The error in the local dynamics}
\label{sec:local_dynamics_err}

We now prove \eqref{eq:OCt_and_barOC_diff}, which we restate in the following lemma:
\begin{lem}
\label{lem:OCt_and_barOC_diff}
Under the same assumptions as in Theorem \ref{thm:err_bound_heisenberg_picture}, we have
\begin{equation*}
    \|O_C(t)-\bar{O}_C^{(r)}\|=\Or\left(\frac{t^2}{r}\right),
\end{equation*}
where $O_C(t)$ and $\bar{O}^{(r)}_C$ are defined in \eqref{eq:OC_dynamics_limit} and \eqref{eq:bar_OC_defn} respectively.
\end{lem}

\begin{proof}

Thanks to the Taylor's theorem, one has
\begin{equation}\label{eq:barOC_taylor}
\bar{O}_C^{(u)} = \bar{O}_C^{(u-1)} + i \tau [H^C_\mathrm{eff}, \bar{O}_C^{(u-1)}] = e^{i H^C_\mathrm{eff} \tau} \bar{O}_C^{(u-1)}  e^{-i H^C_\mathrm{eff} \tau} + \bar{R}^{(u-1)},
\end{equation}
where
\begin{equation}\label{eq:barOC_barR}
    \bar{R}^{(u-1)}  = \int_0^{\tau} e^{i H^C_\mathrm{eff} s} [H^C_\mathrm{eff} ,[H^C_\mathrm{eff} ,\bar{O}_C^{(u-1)} ]] e^{-i H^C_\mathrm{eff} s} (\tau-s)\,ds.
\end{equation}
Here by \eqref{eq:barOC_norm}, one has
\begin{equation}\label{eq:norm_barR}
   \| \bar{R}^{(u-1)} \| \leq \| H^C_\mathrm{eff}\|^2 \tau^2 =  \| H^C_\mathrm{eff}\|^2 \frac{t^2}{r^2}.
\end{equation}
Denote $t_u = u \tau$ with $u = 0, 1, \cdots, r$ so that $t_r = t$. The difference between $\bar{O}_C^{(u)}$ and $O_C(t_u)$ can be written as
\begin{equation*}
    \bar{O}_C^{(u)} - O_C(t_u) = e^{i H^C_\mathrm{eff} \tau} \left( \bar{O}_C^{(u-1)} - O_C(t_{u-1})\right) e^{-i H^C_\mathrm{eff} \tau} 
    + \bar{R}^{(u-1)}. 
\end{equation*}
Taking the norm on both sides, we have
\begin{equation*}
    \|\bar{O}_C^{(u)} - O_C(t_u)\| \leq \|\bar{O}_C^{(u-1)} - O_C(t_{u-1})\| + \| \bar{R}^{(u-1)}\|.
\end{equation*}
It then follows from \eqref{eq:norm_barR} and $\bar{O}_C^{(0)} - O_C(0) = 0$ that
\begin{equation*}
    \|\bar{O}_C^{(r)} - O_C(t)\| \leq \sum_{u = 0}^{r-1}\| \bar{R}^{(u)}\| \leq \| H^C_\mathrm{eff}\|^2 \frac{t^2}{r}.
\end{equation*}
It only remains to show that $\|H^C_\mathrm{eff}\|=\Or(1)$.
Recall that $H^C_\mathrm{eff}$ comes from the effective Hamiltonian $\mathbb{E}[\bar{P}H\bar{P}]$ in \eqref{eq:condition_decoupled_ham}. For each term $E_a$ in $H$, $\bar{P}E_a \bar{P}$ preserves its support because $\bar{P}$ and $E_a$ are both Pauli operators. Therefore
\begin{equation}
H^C_\mathrm{eff} = \sum_{a:\mathrm{Supp} E_a\subset C} \lambda_a \mathbb{E}[\bar{P}E_a\bar{P}].
\end{equation}
From this, and $|\lambda_a|\leq 1$, we have 
\begin{equation}
\|H^C_\mathrm{eff}\|\leq |\{a\in[M]:\mathrm{Supp} E_a\subset C\}| \leq 4^{|C|}\leq 4^k=\Or(1).
\end{equation} 
Therefore we have proved the lemma.
\end{proof}

\section{Reshaping Hamiltonians using Trotterization}
\label{sec:trotter}

Here we consider reshaping the unknown $N$-qubit Hamiltonian $H$ using the second-order Trotter formula. The main idea is the following: In the randomization approach we have constructed an effective Hamiltonian that is a sum of exponentially (in $N$) many terms of the form $PHP$, and here we will consider implementing a similar sum using the 2nd-order Trotter formula. Importantly, this time the sum involves a number of terms that is independent of $N$.

\subsection{Decoupling the dynamics using Trotterization}
\label{sec:decouple_trotter}

First we define a new graph known as a qubit interaction graph.

\begin{defn}[Qubit interaction graph]
\label{defn:qubit_interaction_graph}
First denote $\mathcal{A}_c=\bigsqcup_{C\in\mathcal{V}_c}C$, for $c\in[\chi]$ and $\mathcal{V}_c$ defined in Lemma \ref{lem:coloring_cluster_int_graph}.
The qubit interaction graph corresponding to color $c\in[\chi]$ is defined to be $\mathcal{G}_q^c=(\mathcal{V}_q^c,\mathcal{E}_q^c)$, where $\mathcal{V}_q^c=[N]\setminus\mathcal{A}_c$ contains the qubits that are not contained in $\mathcal{A}_c$, and for any $\alpha,\alpha'\in\mathcal{V}_q^c$ $(\alpha,\alpha')\in\mathcal{E}_q^c$ iff there exists $E_a$ such that $\alpha,\alpha'\in\mathrm{Supp}(E_a)$ and $\mathrm{Supp}(E_a)\cap \mathcal{A}_c\neq \emptyset$.
\end{defn}

We also need to color this graph. The number of colors is given by the following lemma.

\begin{lem}
\label{lem:color_qubit_interaction_graph}
$\mathcal{G}_q^c$ admits a coloring with $\chi_q$ colors. Here $\chi_q\leq (\mathfrak{d}+1)(k-2)+1$ ($\mathfrak{d}$ and $k$ are defined in Definition \ref{defn:low_intersection_ham}).
\end{lem}
\begin{proof}
We will prove that $\mathrm{deg}(\mathcal{G}_q^c)\leq (\mathfrak{d}+1)(k-2)$, and as a result $\chi_q\leq\mathrm{deg}(\mathcal{G}_q^c)+1\leq  (\mathfrak{d}+1)(k-2)+1$.
For each qubit $\alpha$, there are at most $\mathfrak{d}+1$ $E_a$'s such that they act non-trivially on $\alpha$ and on at least one qubit in $\mathcal{A}_c$. Each $E_a$ acts non-trivially on at most $k-1$ other qubits, one of which must be in $\mathcal{A}_c$. Therefore there are at most $(\mathfrak{d}+1)(k-2)$ choices of $\alpha'$ such that $(\alpha,\alpha')\in\mathcal{E}_q^c$.
\end{proof}
We number the colors using $[\chi_q]$. We also denote by $c_q(\alpha)$ the color of qubit $\alpha$.
We now show that, for each color $c$ of the cluster interaction graph $\mathcal{G}$ (Definition \ref{defn:cluster_int_graph}), we can choose Pauli operators $P$ from a set $\mathcal{P}_c$ (to be specified later) of size $4^{\chi_q}$ such that
\begin{equation}
\label{eq:average_dynamics_trotter}
    \frac{1}{|\mathcal{P}_c|}\sum_{P\in\mathcal{P}_c}PHP = \sum_{C\in\mathcal{V}_c}H_C + H_{\mathrm{res}},
\end{equation}
where $H_{\mathrm{res}}$ is supported on $[N]\setminus\mathcal{A}_c$.
We then implement the sum on the right-hand side using the second-order Trotter formula (higher-order formulae will involve evolving backward in time and is therefore unrealistic in our setting). To be more precise, we implement
\begin{equation}
\prod^{\rightarrow}_{1\leq \nu \leq |\mathcal{P}_c|} e^{-iP_{\nu}HP_{\nu} \tau'} \prod^{\leftarrow}_{1\leq \nu \leq |\mathcal{P}_c|} e^{-iP_{\nu}HP_{\nu} \tau'} = \prod^{\rightarrow}_{1\leq \nu \leq |\mathcal{P}_c|} P_{\nu}e^{-iH \tau'}P_{\nu} \prod^{\leftarrow}_{1\leq \nu \leq |\mathcal{P}_c|} P_{\nu}e^{-iH \tau'}P_{\nu},
\end{equation}
where we order the elements in $\mathcal{P}_c$ so that $\mathcal{P}_c=\{P_{\nu}\}$, and $\tau'=\tau/(2|\mathcal{P}_c|)$.
and this will approximate $e^{-i(\sum_{C\in\mathcal{V}_c}H_C+H_{\mathrm{res}})\tau}$ up to second order (with a remainder of order $\tau^3$).

Now we will discuss how to choose $\mathcal{P}_c$. We define $\mathcal{P}_c$ as follows:
\begin{equation}
    \label{eq:defn_Pc}
    \mathcal{P}_c=\Big\{\prod_{\alpha\in[N]\setminus\mathcal{A}_c}\gamma(c_q(\alpha))_{\alpha}:\gamma\in\{I,X,Y,Z\}^{[\chi_q]}\Big\},
\end{equation}
where $\mathcal{A}_c = \bigsqcup_{C\in\mathcal{V}_c} C$ as defined in Definition \ref{defn:qubit_interaction_graph}, and $c_q(\alpha)$ is the color of qubit $\alpha$ in the coloring of the qubit interaction graph.
To see why \eqref{eq:average_dynamics_trotter} is true, let us look at each Pauli terms $E_a$ of $H$. In the first case, if the support of $E_a$ is contained in $\mathcal{A}_c$, then $[P,E_a]=0$ for all $P\in \mathcal{P}_c$. This is because the support of each $P$ does not overlap with $\mathcal{A}_c$ by definition. Consequently $PE_aP=E_a$, and 
\begin{equation}
\label{eq:average_dynamics_preserved_Ea}
    \frac{1}{|\mathcal{P}_c|}\sum_{P\in\mathcal{P}_c}PE_a P = E_a.
\end{equation}
In the second case, if the support of $E_a$ is not contained in $\mathcal{A}_c$, but overlaps with $\mathcal{A}_c$, then we denote $\operatorname{supp}(E_a)\setminus \mathcal{A}_c=\{\alpha_1,\alpha_2,\cdots,\alpha_l\}$ where $l\leq k$. By Definition \eqref{defn:qubit_interaction_graph}, in a coloring of the graph $\alpha_1,\alpha_2,\cdots,\alpha_l$ are all colored differently because they are all linked to each other. Therefore, if we uniformly randomly draw a Pauli operator from $\mathcal{P}_c$, each Pauli operator on $\alpha_1,\alpha_2,\cdots,\alpha_l$ will be chosen independently. From this we can see, just like previously for the randomization method, half of the Pauli operators in $\mathcal{P}_c$ commute with $E_a$ and the other half anti-commute. As a result 
\begin{equation}
\label{eq:average_dynamics_cancelled_Ea}
    \frac{1}{|\mathcal{P}_c|}\sum_{P\in\mathcal{P}_c}PE_a P = 0.
\end{equation}
In the third case, if the support of $E_a$ is disjoint from $\mathcal{A}_c$, then each $PE_a P$ also acts trivially on $\mathcal{A}_c$. We group these terms into the residual term $H_{\mathrm{res}}$.
Combining \eqref{eq:average_dynamics_preserved_Ea} and \eqref{eq:average_dynamics_preserved_Ea} we have \eqref{eq:average_dynamics_trotter}.

\subsection{Isolating the diagonal Hamiltonian using Trotterization}
\label{sec:diagonal_trotter}

For each cluster $C\in\mathcal{V}_c$, we want to learn the diagonal elements of $H_C$ with respect to a Pauli eigenbasis indexed by $\gamma_C\in\{0,x,y,z\}^C$, as defined in Definition \ref{defn:Pauli_eigenbasis}. 

To this end, for a set of Pauli eigenbases indexed by $\{\gamma_C\}_{C\in\mathcal{V}_c}$, we define
\begin{equation}
    \label{eq:defn_Qcgamma}
    \mathcal{Q}_c = \Big\{\prod_{C\in\mathcal{V}_c}\prod_{\alpha\in C}(\gamma_C(\alpha)_{\alpha})^{b_{\zeta_C(\alpha)}}:b_1,b_2,\cdots,b_{k}\in\{0,1\}\Big\},
\end{equation}
where $\zeta_C:C\to[|C|]$ is an arbitrary fixed ordering of $C$.
Then we will have
\begin{equation}
\label{eq:average_dynamics_trotter_diag}
    \frac{1}{|\mathcal{Q}_c||\mathcal{P}_c|}\sum_{Q\in \mathcal{Q}_c}\sum_{P\in\mathcal{P}_c}QPHPQ = \sum_{C\in\mathcal{V}_c}H^C_{\mathrm{diag}}(\gamma_C)+H_{\mathrm{res}},
\end{equation}
where 
\begin{equation}
    H^C_{\mathrm{diag}}(\gamma_C) = \frac{1}{|\mathcal{Q}_c|}\sum_{Q\in \mathcal{Q}_c} QH_C Q.
\end{equation}
is the diagonal of the Hamiltonian $H_C$ with respect to the Pauli eigenbasis indexed by $\gamma_C$. It is the same Hamiltonian as given in \eqref{eq:new_effective_ham_diag}.
Therefore, to extract the diagonal Hamiltonian, we can implement
\begin{equation}
\label{eq:defn_U_tau}
    \begin{aligned}
    \mathcal{U}(\tau) &=\prod^{\rightarrow}_{1\leq \nu\leq |\mathcal{R}_c|} e^{-iP_{\nu}HP_{\nu} \tau'} \prod^{\leftarrow}_{1\leq \nu\leq |\mathcal{R}_c|} e^{-iP_{\nu}HP_{\nu} \tau'} \\
    &= \prod^{\rightarrow}_{1\leq \nu\leq |\mathcal{R}_c|} P_{\nu}e^{-iH \tau'}P_{\nu} \prod^{\leftarrow}_{1\leq \nu\leq |\mathcal{R}_c|} P_{\nu}e^{-iH \tau'}P_{\nu} \\
    &\approx e^{-i(\sum_{C\in\mathcal{V}_c}H^C_{\mathrm{diag}}(\gamma_C)+H_{\mathrm{res}})\tau} = e^{-iH_{\mathrm{res}}\tau}\prod_{C\in\mathcal{V}_c} e^{-iH^C_{\mathrm{diag}}(\gamma_C)\tau},
    \end{aligned}
\end{equation}
where $\mathcal{R}_c=\{PQ:P\in\mathcal{P}_c,Q\in\mathcal{Q}_c\}=\{P_{\nu}\}$, $\tau'=\tau/(2|\mathcal{R}_c|)$, and in $\approx$ we neglected all terms that are of order $\tau^3$ or higher.

The number of Pauli operators needed  to implement a step for a short time $\tau$ scale linearly with $|\mathcal{R}_c|=|\mathcal{P}|_c|\mathcal{Q}_c|$. Because, by Lemma \ref{lem:color_qubit_interaction_graph},
\begin{equation}
    |\mathcal{P}_c|= 4^{\chi_q}\leq 4^{(\mathfrak{d}+1)(k-2)+1},\quad |\mathcal{Q}_c|\leq 2^k,
\end{equation}
we have 
\begin{equation}
    |\mathcal{R}_c|\leq 4^{(\mathfrak{d}+1)(k-2)+k/2+1}.
\end{equation}
It is important to note that $|\mathcal{R}_c|$ is independent of the system size $N$.

Below we estimate how many Trotter steps are needed to make the actual dynamics close to limiting dynamics. The proof of this theorem is given in Section \ref{sec:deviate_limit_dynamics_trotter}.

\begin{thm}[Number of Trotter steps needed]
\label{thm:num_trotter_steps_needed}
Assume that $H$ is a low-intersection Hamiltonian defined in Definition \ref{defn:low_intersection_ham}, and $\mathcal{V}=\bigsqcup_{c\in[\chi]}\mathcal{V}_c$ is a coloring according to Lemma \ref{lem:coloring_cluster_int_graph}, $c\in[\chi]$, and $\gamma_C\in\{X,Y,Z\}^C$ for each $C\in\mathcal{V}_c$. 

Let $\mathcal{U}(\tau)$ be defined in \eqref{eq:defn_U_tau}, and let $\rho(t)=\mathcal{U}(\tau)^r\rho(0)(\mathcal{U}(\tau)^{\dagger})^r$ be the state of the quantum system at time $t$ after being initialized in state $\rho(0)$.
Then there exists $r_0=\Or(t^{3/2}/\varepsilon^{1/2})$ such that for any $r>r_0$, such that for any $C$ and $O_C$ supported on $C$, with $\|O_C\|\leq 1$, we have
\begin{equation}
    \Big|\Tr[(O_C\otimes I)\rho(t)]-\Tr[O_C e^{-iH_{\mathrm{diag}}^C(\gamma_C) t}\rho_C e^{iH_{\mathrm{diag}}^C(\gamma_C) t}]\Big|\leq \varepsilon,
\end{equation}
where $\rho_C=\Tr_{[N]\setminus C}\rho(0)$.
\end{thm}

In our Hamiltonian learning algorithm, we only need to ensure that the actual dynamics deviate from the limiting dynamics by a small constant. Therefore it suffices to choose $r=\Or(\epsilon^{-3/2})$ in the above theorem ($\epsilon$ is the precision for Hamiltonian parameters, and $\epsilon^{-1}$ is the evolution time needed for robust phase estimation), as opposed to $r=\Or(\epsilon^{-2})$ needed in the randomization approach. We summarize the costs of the Trotter-based approach in the following theorem
\begin{thm}[Learning many-body Hamiltonian by reshaping with Trotter formula]
    \label{thm:ham_learn_Trotter}
    Assume that $H$ is a low-intersection Hamiltonian defined in Definition \ref{defn:low_intersection_ham}. Then we can generate estimates $\{\hat{\lambda}_a\}$ for parameters $\{\lambda_a\}$ in \eqref{eq:Ham_gen_low_intersect}, such that
    \begin{equation}
    \Pr[|\hat{\lambda}_a-\lambda_a|>\epsilon]<\delta
    \end{equation}
    for all $a\in[M]$ with the following cost:
    \begin{enumerate}
        \item $\Or(\epsilon^{-1}\log(\delta^{-1}))$ total evolution time;
        \item $\Or(\polylog(\epsilon^{-1})\log(\delta^{-1}))$ number of experiments;
        \item $\Or(N\epsilon^{-3/2}\polylog(\epsilon^{-1})\log(\delta^{-1}))$ single-qubit Clifford gates.
    \end{enumerate}
     Moreover, this algorithm is robust against SPAM error.
\end{thm}
The SPAM-robustness follows in a similar way as in the proof of Theorem \eqref{thm:ham_learn_partial_pauli_twirling}.

\section{Deviation from the limiting dynamics in Trotterization}
\label{sec:deviate_limit_dynamics_trotter}

In this section we will prove Theorem \ref{thm:num_trotter_steps_needed}. Following the discussion in Section \ref{sec:deviate_limit_dynamics}, we only need to prove the following theorem, providing an error bound for the evolution of an arbitrary local operator in the Heisenberg picture.
\begin{thm}
    \label{thm:err_bound_heisenberg_picture_trotter}
    We assume that $H$ is a low-intersection Hamiltonian as defined in Definition \ref{defn:low_intersection_ham}, let $\mathcal{U}(\tau)$ be as defined in \eqref{eq:defn_U_tau}, where the set of Pauli operators $\mathcal{R}_c$ satisfies
    \begin{equation}
    \label{eq:condition_decoupled_ham_trotter}
        \frac{1}{|\mathcal{R}_c|}\sum_{P\in\mathcal{R}_c}PHP = H_{\mathrm{eff}}^C + H_{\mathrm{env}},
    \end{equation}
    where $H_{\mathrm{eff}}^C$ is supported on a subsystem $C$ ($|C|=\Or(1)$) and $H_{\mathrm{env}}$ is supported on the rest of the system.
    Then
    \begin{equation}
    \label{eq:heisenberg_picture_err_bound_final_time_trotter}
        \Big\|\left(\mathcal{U}(\tau)^\dagger \right)^r \left(  O_C \otimes I \right) \mathcal{U}(\tau)^r - e^{iH^C_{\mathrm{eff}}t}O_C e^{-iH^C_{\mathrm{eff}}t}\otimes I \Big\| = \Or(t^3/r^2),
    \end{equation}
    for any $O_C$ supported on $C$ satisfying $\|O_C\|\leq 1$. In particular, the constant in $\Or(t^3/r^2)$ does not depend on the system size $N$ or the number of terms $M$.
\end{thm}

\begin{proof}
We rewrite the $\mathcal{U}(\tau)$ defined in \eqref{eq:defn_U_tau} as $\mathcal{U}(\tau,0)$, and also write it as the unitary time evolution operator due to a time-dependent Hamiltonian:
\begin{equation}
     \mathcal{U}(\tau,0) = \mathcal{T} e^{-i \int_0^\tau \widetilde{H}(s) \, ds},
\end{equation}
where $\widetilde{H}(s)$ is a piecewise time-dependent Hamiltonian and is defined as follows: we divide the interval $[0, \tau]$ into $2|\mathcal{R}_c|$ sub-intervals, and one each sub-interval $\widetilde{H}$ equals $P H P$ with $P\in\mathcal{R}_c$ and then reverse the order.

Notice that, in order to prove \eqref{eq:heisenberg_picture_err_bound_final_time_trotter}, because the long time error grows linearly with respect to $r$ thanks to the triangle inequality and the unitarity of both underlying dynamics, it is sufficient to bound the one-step error (i.e., the local truncation error using the terminology of numerical analysis \cite{Leveque2007finite})
\begin{equation} \label{trotter_short_time_ob_err}
\norm{ \left( e^{iH_{\mathrm{eff}}^C \tau}  O_C e^{-iH_{\mathrm{eff}}^C \tau} \right) \otimes I -
\mathcal{U}(\tau,0)^\dagger \left(  O_C \otimes I \right) \mathcal{U}(\tau,0)  } = \norm{T_{\rm eff}(\tau)  - T_{\rm tr} (\tau) },
\end{equation}
where
\begin{equation}
T_{\rm eff}(\tau) = \left( e^{iH_{\mathrm{eff}}^C \tau}  O_C e^{-iH_{\mathrm{eff}}^C \tau} \right) \otimes I,\quad T_{\rm tr} (\tau) = \mathcal{U}(\tau,0)^\dagger \left(  O_C \otimes I \right) \mathcal{U}(\tau,0).
\end{equation}

We start by performing series expansion of both terms on the left-hand side in \eqref{trotter_short_time_ob_err}. By the Taylor's theorem, we have
\begin{align*}
    & e^{iH_{\mathrm{eff}}^Ct} O_C e^{-iH_{\mathrm{eff}}^Ct}
    \\
    = & O_C + i \tau [H_{\mathrm{eff}}^C, O_C] 
    - \frac{\tau^2}{2} [H_{\mathrm{eff}}^C, [H_{\mathrm{eff}}^C, O_C] 
    - i \int_0^\tau \frac{s^2}{2} e^{iH_{\mathrm{eff}}^Cs}
    [H_{\mathrm{eff}}^C, [H_{\mathrm{eff}}^C, [H_{\mathrm{eff}}^C, O_C]]  e^{-iH_{\mathrm{eff}}^Ct}\, ds.
\end{align*}
For the second term, a key observation is that 
\begin{equation}\label{eq:T_tr_time_ordering}
    T_{\rm tr}(\tau) = \mathcal{T} e^{i \int_0^\tau \ad_{H(t-s)} \, ds} \left( O_C \otimes I \right).
\end{equation} 
To see this, denote $O_C \otimes I$ as $O$ and
$F(t,s) : = \mathcal{U}(s,t) O \mathcal{U}(t,s) $, it follows from taking the derivative of $F(t,s)$ with respect to $s$ that
\begin{equation*}
    \partial_s F(s,t) = - i [H(s), F(s,t)] = - i \ad_{H(s)} F(s,t), \quad F(s = t, t) = O,
\end{equation*}
so that 
\begin{equation*}
    \partial_s F(t-s, t) =  i [H(t-s), F(t-s, t)] =  i \ad_{H(t-s)} F(t-s, t), \quad F(t-s, t)|_{s = 0} = O.
\end{equation*}
We now perform the Dyson series expansion to \eqref{eq:T_tr_time_ordering} and arrive at
\begin{align*}
    T_{\rm tr}(\tau) 
    & = \sum_{N = 0}^\infty
    \sum i^N \int_0^\tau d t_1 \int_0^{t_1} d t_2 \cdots \int_0^{t_{N-1}} d t_n
    \ad_{H(t-t_1)} \circ  \ad_{H(t-t_2)} \circ \cdots \circ \ad_{H(t-t_n)} (O)
    \\
    & = \sum_{N = 0}^\infty i^N \int_0^{\tau} d s_1 \int_{s_1}^t d s_2 \cdots \int_{s_{N-1}}^t ds_n
    [H(s_1), [H(s_2), \cdots, [H(s_n), O] \cdots ] ].
\end{align*}
Gathering terms of $\Or(\tau^0)$, one has $O_C \otimes I$. The terms of $\Or(\tau^1)$ read
\begin{equation*}
    i \int_0^\tau d s_1 [H(s_1), O_C \otimes I] = i  \left[\frac{\tau}{|\mathcal{R}_c|} \sum_{P \in \mathcal{R}_c} P HP , O_C \otimes I \right] 
    = i\tau [H_{\mathrm{eff}}^C, O_C] \otimes I.
\end{equation*}
The terms of $\Or(\tau^2)$ are
\begin{align*}
     &-\int_0^t d s_1 \int_{s_1}^t d s_2 [H(s_1), [H(s_2), O_C \otimes I ]
    = - \frac{\tau^2}{4 |\mathcal{R}_c|^2} \left[ \sum_{j = 1}^{2|\mathcal{R}_c|} H_j, \left[\sum_{l = j}^{|\mathcal{R}_c|} H_l, O_C \otimes I  \right] \right]
    \\
    =& - \frac{\tau^2}{4 |\mathcal{R}_c|^2} \left( \left[ \sum_{j = 1}^{|\mathcal{R}_c|} H_j, \left[2\sum_{l = 1}^{|\mathcal{R}_c|} H_l - \sum_{l = 1}^j H_l, O_C \otimes I  \right] \right]
    + \left[ \sum_{j = |\mathcal{R}_c|}^{2|\mathcal{R}_c|} H_j, \left[\sum_{l = 1}^j H_l, O_C \otimes I  \right] \right] 
    \right)
    \\
    =& - \frac{\tau^2}{2} [H_{\mathrm{eff}}^C, [H_{\mathrm{eff}}^C, O_C]] \otimes I
    ,
\end{align*}
where we used the fact that $H(s)$ is piece-wise constant, and we label its value on each piece as $H_j$ so that $H_{j+|\mathcal{R}_c|} =H_{|\mathcal{R}_c| - j}$ for $1 \leq j \leq |\mathcal{R}_c|$.
It can be seen that the first three terms of $T_{\rm tr}(\tau)$ match those of $T_{\rm eff}(\tau)$.
For the terms with $j\geq 3$, we note that for each $s$, $\widetilde{H}(s)=PHP$ for some Pauli operator $P$, and $\widetilde{H}(s)$ therefore consists of Pauli operators that have exactly the same supports as those in $H$. Consequently, using the same argument as \eqref{eq:j_nested_comm}-\eqref{eq:upper_bound_j_nested_comm}, each term can be bounded through
\begin{equation}
\|[H(s_1), [H(s_2), \cdots, [H(s_n), O] \cdots ] ]\|\leq j!(2(\mathfrak{d}+1))^j.
\end{equation}
As a result the sum of these terms is bounded by $\Or(\tau^3)$.
Also note that the last term of \eqref{eq:T_tr_time_ordering} can be bounded by $4\norm{H_{\mathrm{eff}}^C}^3\tau^3/3$, where $H_{\mathrm{eff}}^C=\Or(1)$ as argued in the proof of Lemma \ref{lem:effective_ham_isolate_diagonal}. Therefore, we can conclude that 
\begin{equation*}
    \norm{T_{\rm eff}(\tau)  - T_{\rm tr} (\tau) } \leq A_5 \tau^3,
\end{equation*}
for some constant $A_5$ independent of the system size,
and hence
\begin{equation} \label{trotter_ob_err_long_time_result}
\norm{ \left( e^{iH_{\mathrm{eff}}^Ct} O_C e^{-iH_{\mathrm{eff}}^Ct} \right) \otimes I -
\left(\mathcal{U}(\tau,0)^\dagger \right)^r  \left(  O_C \otimes I \right) \left(\mathcal{U}(\tau,0) \right)^r } \leq  A_5 r\tau^3 = A_5 \frac{t^3}{r^2},
\end{equation}
which establishes the claim of this theorem.
\end{proof}

\section{Lower bound for learning Hamiltonian from dynamics}
\label{sec:lower_bound}

In this section, we present a fundamental lower bound on the total evolution time for any learning algorithm that tries to learn an unknown Hamiltonian from dynamics.

\subsection{Model of quantum experiments}

We consider a unitary $U(t)$ parameterized by time $t$ that implements $e^{-i H t}$ for an unknown $N$-qubit Hamiltonian $H$.
A learning agent can access $U(t)$ by quantum experiments.
We define a single ideal quantum experiment as follows.
The definition resembles the formalism given in \cite{huang2022foundations}.

\begin{defn}[A single ideal experiment]
Given an unknown $N$-qubit unitary $U(t) = e^{-i H t}$ parameterized by time $t$.
A single ideal experiment $E^{(0)}$ is specified by:
\begin{enumerate}
    \item an arbitrary $N'$-qubit initial state $\ket{\psi_0} \in \mathbb{C}^{2^{N'}}$ with an integer $N' \geq N$,
    \item an arbitrary POVM $\mathcal{F} = \{ M_i \}_i$ on $N'$-qubit system,
    \item an $N'$-qubit unitary of the following form,
    \begin{equation}
    U_{K+1} (U(t_K) \otimes I) U_K \ldots U_3 (U(t_2) \otimes I) U_2 (U(t_1) \otimes I) U_1,
    \end{equation}
    for some arbitrary integer $K$, arbitrary evolution times $t_1, \ldots, t_K \in \mathbb{R}$, and arbitrary $N'$-qubit unitaries $U_1, \ldots, U_K, U_{K+1}$.
    Here $I$ is the identity unitary on $N'-N$ qubits.
\end{enumerate}
A single run of $E^{(0)}$ returns an outcome from performing the POVM $\mathcal{F}$ on the state
\begin{equation}
    U_{K+1} (U(t_K) \otimes I) U_K \ldots U_3 (U(t_2) \otimes I) U_2 (U(t_1) \otimes I) U_1 \ket{\psi_0}.
\end{equation}
The evolution time of the experiment is defined as $t(E^{(0)}) \triangleq \sum_k |t_k|$.
\end{defn}

The learning algorithm can adaptively choose each quantum experiment based on past measurement outcomes.
We consider the quantum experiments to have a small unknown state preparation and measurement (SPAM) error.
Given an initial state $\ket{\psi_0}$, the actual initial state being prepared on the quantum system is $\rho_{0}$, which is equal to $\ketbra{\psi_0}{\psi_0}$ up to a small constant error $\eta$ in the trace norm.
Similarly, given a POVM $\mathcal{F} = \{ M_i \}_i$, the actual POVM being measured is $\tilde{\mathcal{F}} = \{ \tilde{M}_i \}_i$, where $\tilde{M}_i$ is equal to $M_i$ up to a small constant error $\eta$ in the trace norm.
A small SPAM error is always present in any quantum experiment.
Measurement noises are particularly assured as measurements require the quantum system to interact with the macroscopic classical world, which often result in decoherence.
A practically useful algorithm for characterizing and benchmarking quantum systems \cite{knill2008randomized, magesan2011scalable, KimmelLowYoder2015robust, erhard2019characterizing, harper2020efficient, elben2022randomized, huang2022foundations} has to be robust against a small amount of SPAM error.

If a learning algorithm works for any small unknown SPAM error, then the learning algorithm must apply to experiments where only the measurement is subject to a small depolarizing noise $\eta$.
We give a single experiment with measurement noise $\eta$ in the following definition.
The lower bound will be proved assuming access to the experiments with a small measurement noise $\eta = \Theta(1)$.

\begin{defn}[A single experiment with measurement noise]
A single experiment $E^{(\eta)}$ with measurement noise $\eta = \Theta(1)$ is specified by the same parameters as a single ideal experiment $E^0$.
Given the ideal POVM $\mathcal{F} = \{M_i\}_i$, the measurement outcome of $E^{(\eta)}$ is obtained by performing the noisy POVM $\mathcal{F}^{(\eta)} = \{ (1 - \eta) M_i + \eta \Tr(M_i) (I / 2^{N'}) \}$ on the state
\begin{equation}
    U_{K+1} (U(t_K) \otimes I) U_K \ldots U_3 (U(t_2) \otimes I) U_2 (U(t_1) \otimes I) U_1 \ket{\psi_0}.
\end{equation}
The evolution time of the experiment is defined as $t(E^{(\eta)}) \triangleq \sum_k |t_k|$.
\end{defn}

We formally define a learning algorithm with total evolution time $T$ as follows.

\begin{defn}[Learning algorithm with bounded total evolution time]
Given $T > 0, 0.5 > \eta > 0$.
A learning algorithm with total evolution time $T$ and measurement noise $\eta$ can obtain measurement outcomes from an arbitrary number of experiments $E^{(\eta)}_1, E^{(\eta)}_2, \ldots$ as long as
\begin{equation}
    \sum_{i} t(E^{(\eta)}_i) \leq T.
\end{equation}
The parameters specifying each experiment $E^{(\eta)}_i$ can depend on the measurement outcomes from previous experiments $E^{(\eta)}_1, \ldots, E^{(\eta)}_{i-1}$.
\end{defn}

\subsection{Learning task and lower bound}

After defining the learning algorithm and the possible sets of experiments, we are ready to state the lower bound on the total evolution time required to learn an $N$-qubit Hamiltonian from dynamics.
The theorem is stated as follows.
This scaling matches that of our proposed learning algorithm.

\begin{thm} \label{thm-lowerbound}
Given two integers $N, M$, two real values $\epsilon, \delta \in (0, 1)$, and a set $\{E_1, \ldots, E_M\} \subseteq \{I, X, Y, Z\}^{\otimes N} \setminus \{I^{\otimes N}\}$ of $M$ Pauli operators.
Consider any learning algorithm with a total evolution time $T$ and a constant measurement noise $\eta \in (0, 0.5)$, such that for any $N$-qubit Hamiltonian $H = \sum_{a=1}^M \lambda_a E_a$ with unknown parameters $|\lambda_a| \leq 1$, after multiple rounds of experiments, the algorithm can estimate $\lambda_a$ to $\epsilon$-error with probability at least $1 - \delta$ for any $a \in \{1, \ldots, M\}$.
Then
\begin{equation}
    T \geq \frac{\log(1 / 2 \delta)}{2 \epsilon \log(1 / \eta)} = \Omega\left( \frac{\log(1 / \delta)}{\epsilon} \right).
\end{equation}
\end{thm}

\noindent Even when the measurement noise $\eta = 10^{-10}$, we still have $T \geq \log(1 / 2\delta) / (40 \epsilon)$.

\subsection{Proof of Theorem~\ref{thm-lowerbound}}

The proof of the lower bound is separated into four parts.
The first part in Section~\ref{sec:reduction} reduces the learning problem to a binary distinguishing task.
The second part in Section~\ref{sec:TV-one} provides an upper bound for the total variation (TV) distance between the distribution over measurement outcomes under a single experiment.
The third part in Section~\ref{sec:TV-multi} uses a learning tree representation described in \cite{huang2022quantum, chen2022exponential} and provides an upper bound for the total variation distance between distribution over the leaf nodes of the tree.
The fourth part in Section~\ref{sec:finalize-lowerbd} utilizes LeCam's two-point method to turn the TV upper bounds into a lower bound for the total evolution time.

\subsubsection{Reduction}
\label{sec:reduction}

If a learning algorithm can achieve the original learning task considered in  Theorem~\ref{thm-lowerbound}, then it could solve a simpler learning task, where the unknown Hamiltonian $H$ can only be one of the following two choices.
The unknown $N$-qubit Hamiltonian $H$ is either $\epsilon E_1$ or $-\epsilon E_1$ with equal probability, where $E_1 \in \{I, X, Y, Z\}^{\otimes N} \setminus \{I^{\otimes N}\}$ is an $N$-qubit Pauli operator that is not an identity operator.
We denote $U_{\pm}(t)$ to be the unitary corresponding to evolution under the two Hamiltonians.
If there is a learning algorithm with a total evolution time at most $T$ that succeeds in the learning task stated in Theorem~\ref{thm-lowerbound}, then we can use the learning algorithm to successfully distinguish between $\pm \epsilon E_1$ with probability at least $1 - \delta$.
Hence, a lower bound on $T$ for this simpler learning task immediately implies a lower bound on $T$ for the original learning task.

We can characterize the diamond distance between the two unitaries $U_{\pm}(t)$.
For a unitary $U$, we consider $\mathcal{U}(\rho) = U \rho U^\dagger$ to be the corresponding quantum channel (CPTP map).

\begin{lem}[Diamond distance between $U_{\pm}(t)$] \label{lem:diamond-Upm}
$\norm{\mathcal{U}_{+}(t) - \mathcal{U}_-(t)}_\diamond \leq 4 \epsilon |t|.$
\end{lem}
\begin{proof}
The spectrum of $U_{+}(t)^\dagger U_-(t)$ is given by $e^{i 2 \epsilon t}, e^{-i 2 \epsilon t}$.
From \cite{nechita2018almost, watrous2018theory}, we have
\begin{equation}
 \norm{\mathcal{U}_{+}(t) - \mathcal{U}_-(t)}_\diamond = 2 \sin (2 \epsilon |t|)   
\end{equation}
if $2 \epsilon |t| < \pi / 2$, otherwise we have
\begin{equation}
    \norm{\mathcal{U}_{+}(t) - \mathcal{U}_-(t)}_\diamond = 2.
\end{equation}
In both cases, we have $\norm{\mathcal{U}_{+}(t) - \mathcal{U}_-(t)}_\diamond \leq 4 \epsilon |t|.$
\end{proof}

\subsubsection{TV upper bound for a single experiment}
\label{sec:TV-one}

We begin by proving the upper bound on total variation distance for a single quantum experiment.

\begin{lem}[TV for one experiment] \label{lem:TV-one-exp}
Given an unknown unitary $U(t)$ equal to either $U_+(t)$ or $U_-(t)$, and a single experiment $E^{(\eta)}$ with measurement noise $\eta$ specified by the following parameters,
\begin{enumerate}
    \item an arbitrary $N'$-qubit initial state $\ket{\psi_0} \in \mathbb{C}^{2^{N'}}$ with an integer $N' \geq N$,
    \item an arbitrary POVM $\mathcal{F} = \{ M_i \}_i$ on $N'$-qubit system,
    \item an $N'$-qubit unitary of the following form,
    \begin{equation}
    U_{K+1} (U(t_K) \otimes I) U_K \ldots U_3 (U(t_2) \otimes I) U_2 (U(t_1) \otimes I) U_1,
    \end{equation}
    for some arbitrary integer $K$, arbitrary evolution times $t_1, \ldots, t_K$, and arbitrary $N'$-qubit unitaries $U_1, \ldots, U_K, U_{K+1}$.
    Here $I$ is the identity unitary on $N'-N$ qubits.
\end{enumerate}
Let $p_\pm(i)$ be the probability of obtaining the measurement outcome $i$ by performing $\mathcal{F}^{(\eta)} = \{ (1 - \eta) M_i + \eta \Tr(M_i) (I / 2^{N'}) \}$ on the output state when $U(t) = U_\pm(t)$.
Then
\begin{equation}
    \mathrm{TV}(p_+, p_-) \leq (1 - \eta) \min(2 \epsilon t(E^{(\eta)}), 1),
\end{equation}
where $t(E^{(\eta)}) = \sum_{k=1}^K |t_k|$ is the total evolution time in this single experiment $E^{(\eta)}$.
\end{lem}
\begin{proof}
We define $\ket{\psi_{\pm}} = U_K U_{\pm}(t_K) \ldots U_3 U_{\pm}(t_2) U_2 U_{\pm}(t_1) U_1 \ket{\psi_0}$.
By triangle inequality and telescoping sum, we have the following upper bound on the trace distance,
\begin{equation}
    \norm{\ketbra{\psi_+}{\psi_+} - \ketbra{\psi_-}{\psi_-}}_1 \leq \sum_{k=1}^K \norm{U_{+}(t_k) - U_-(t_k)}_\diamond \leq 4 \epsilon \sum_{k=1}^K \left|t_k\right| = 4 \epsilon t(E^{(\eta)}).
\end{equation}
The second inequality follows from Lemma~\ref{lem:diamond-Upm}.
We can now upper bound the total variation distance for the classical probability distribution when we measure the final state using the ideal POVM measurement $\mathcal{F} = \{ M_i \}_i$,
\begin{equation}
    \frac{1}{2} \sum_{i} \left| \bra{\psi_+} M_i \ket{\psi_+} - \bra{\psi_-} M_i \ket{\psi_-} \right| \leq \frac{1}{2} \norm{\ketbra{\psi_+}{\psi_+} - \ketbra{\psi_-}{\psi_-}}_1 \leq 2 \epsilon t(E^{(\eta)}).
\end{equation}
Because the total variation distance is upper bounded by $1$, we have
\begin{equation}
    \frac{1}{2} \sum_{i} \left| \bra{\psi_+} M_i \ket{\psi_+} - \bra{\psi_-} M_i \ket{\psi_-} \right| \leq \min(2 \epsilon t, 1).
\end{equation}
When we measure using the noisy POVM $\mathcal{F}^{(\eta)} = \{ \tilde{M}_i =  (1 - \eta) M_i + \eta \Tr(M_i) (I / 2^{N'}) \}$ instead of $\mathcal{F}$,
the total variation distance between the measurement outcome distribution is
\begin{align} \label{eq:single-exp-lb}
    &\frac{1}{2} \sum_{i} \left| \bra{\psi_+} \tilde{M}_i \ket{\psi_+} - \bra{\psi_-} \tilde{M}_i \ket{\psi_-} \right|\\
    &= \frac{1}{2} (1 - \eta) \sum_{i} \left| \bra{\psi_+} M_i \ket{\psi_+} - \bra{\psi_-} M_i \ket{\psi_-} \right| \leq (1 - \eta) \min(2 \epsilon t(E^{(\eta)}), 1).
\end{align}
By definition, we have $p_{\pm}(i) = \bra{\psi_\pm} \tilde{M}_i \ket{\psi_\pm}$.
Hence, $\mathrm{TV}(p_+, p_-) \leq (1 - \eta) \min(2 \epsilon t(E^{(\eta)}), 1)$, which is the total variation distance between the measurement outcome distribution over the two Hamiltonians under a single experiment.
\end{proof}

\subsubsection{TV upper bound for many experiments}
\label{sec:TV-multi}

To handle adaptivity in the choice of experiments,
we consider the rooted tree representation $\mathcal{T}$ described in \cite{huang2022quantum, chen2022exponential}.
Each node in the tree corresponds to the sequence of measurement outcomes the algorithm has seen so far.
We can also think of the node as the memory state of the algorithm.
At each node $v$, the algorithm runs a single experiment $E^{(\eta)}_v$ with measurement noise $\eta$ specified by
\begin{enumerate}
    \item an arbitrary $N'_v$-qubit initial state $\ket{\psi_{v, 0}} \in \mathbb{C}^{2^{N'_v}}$ with an integer $N'_v \geq N$,
    \item an arbitrary POVM $\mathcal{F}_v = \{ M_{v, i} \}_{i=1}^{L_v}$ with $L_v$ outcomes on $N'_v$-qubit system,
    \item an $N'_v$-qubit unitary of the following form,
    \begin{equation}
    U_{v, K_v+1} (U(t_{v, K_v}) \otimes I) U_{v, K_v} \ldots U_{3, v} (U(t_{v, 2}) \otimes I) U_{v, 2} (U_{\pm}(t_{v, 1}) \otimes I) U_{v, 1},
    \end{equation}
    for some arbitrary integer $K_v$, arbitrary evolution times $t_{v, 1}, \ldots, t_{v, K_v} \in \mathbb{R}$, and arbitrary $N'_v$-qubit unitaries $U_{v, 1}, \ldots, U_{v, K+1}$.
    Here $I$ is the identity unitary on $N'_v-N$ qubits.
\end{enumerate}
Each experiment $E^{(\eta)}_v$ produces a measurement outcome $i \in \{1, \ldots, L_v\}$, which moves the algorithm from the node $v$ to one of its child node.
At a leaf node $\ell$, the algorithm stops.
By considering the rooted tree representation and allowing the experiment to depend on each node in the tree, we cover all possible learning algorithm that can adaptively choose the experiment that it runs based on previous measurement outcomes.

For each node $v$ on tree $\mathcal{T}$, we denote $p^{(\mathcal{T})}_\pm(v)$ to be the probability of arriving at the node $v$ in the experiments when the unknown unitary $U(t) = U_{\pm}(t)$ and the algorithm begins from the root of $\mathcal{T}$.
We can establish the following total variation upper bound.

\begin{lem}[TV for multiple experiments] \label{lem:tv-mult}
Consider a rooted tree representation $\mathcal{T}$ for a learning algorithm with total evolution time $T$ and measurement noise $\eta \in (0, 0.5)$.
We have
\begin{equation}
    \mathrm{TV}(p_+^{(\mathcal{T})}, p_-^{(\mathcal{T})}) \leq 1 - \eta^{2 \epsilon T},
\end{equation}
which is an upper bound for the total variation of the outcomes under multiple experiments.
\end{lem}
\begin{proof}
For each node $v$, we give the following definitions,
\begin{itemize}
    \item $\mathcal{T}_v$ is the subtree with root $v$.
    \item $p_\pm^{(v)}$ is the distribution over the child nodes of $v$ by considering the probability of moving from $v$ to that child node under the unknown unitary $U_{\pm}(t)$.
    \item $p_\pm^{(\mathcal{T}_v)}$ is the distribution over the leaf nodes for subtree $\mathcal{T}_v$ by considering the probability of ending at that leaf node starting from node $v$ under the unknown unitary $U_{\pm}(t)$.
    \item $t^{(v)} \triangleq t(E^{(\eta)}_v) \geq 0$ is the evolution time for the single experiment $E^{(\eta)}_v$.
    \item $t(\mathcal{T}_v)$ is the maximum of the sum of the evolution time over all paths from root $v$ of the subtree $\mathcal{T}_v$ to a leaf node of $\mathcal{T}_v$,
    \begin{equation}
        t(\mathcal{T}_v) = \max_{P: \mathrm{path}\,\mathrm{on}\,\mathcal{T}_v} \sum_{w \in P} t(E^{(\eta)}_w).
    \end{equation}
    Because the total evolution time of the learning algorithm is upper bounded by $T$, the total evolution time of the full tree $\mathcal{T}$ satisfies $t(\mathcal{T}) \leq T$.
\end{itemize}
We will prove this lemma by an induction over the subtree of $\mathcal{T}$.
The inductive hypothesis is given as follows. For any subtree $\mathcal{T}_v$ with root $v$,
\begin{equation}
    1 - \mathrm{TV}(p_+^{(\mathcal{T}_v)}, p_-^{(\mathcal{T}_v)}) \geq \eta^{2 \epsilon t(\mathcal{T}_v)}.
\end{equation}
The base case is when $v$ is a leaf node.
At the leaf node $\ell$, we have $\mathrm{TV}(p_+^{(\mathcal{T}_\ell)}, p_-^{(\mathcal{T}_\ell)}) = 0$ and $t(\mathcal{T}_\ell) = 0$. Hence, the induction hypothesis holds.

To prove the inductive step, we define $\mathrm{child}(v)$ the be the set of child node of $v$ and recall the following identity on two probability distributions $p_{\pm}$ over a set $\mathcal{X}$,
\begin{equation}
    1 - \mathrm{TV}(p_+, p_-) = \sum_{x \in \mathcal{X}} \min\big( p_+(x), p_-(x) \big).
\end{equation}
We can obtain a lower bound on the failure probability for the node $v$ as follows,
\begin{align}
    &1 - \mathrm{TV}(p_+^{(\mathcal{T}_v)}, p_-^{(\mathcal{T}_v)}) \label{eqn:tv-mult-1} \\
    &= \sum_{\substack{\ell \in \mathrm{leaf}( \mathcal{T}_v ) }} \min\left( p_+^{(\mathcal{T}_v)}(\ell), p_-^{(\mathcal{T}_v)}(\ell) \right) \\
    &= \sum_{w \in \mathrm{child}(v)} \sum_{\substack{\ell \in \mathrm{leaf}( \mathcal{T}_w ) }} \min\left( p^{(v)}_+(w) p_+^{(\mathcal{T}_w)}(\ell), p^{(v)}_-(w) p_-^{(\mathcal{T}_w)}(\ell) \right) \\
    &\geq \sum_{w \in \mathrm{child}(v)} \min\left( p^{(v)}_+(w), p^{(v)}_-(w) \right) \sum_{\substack{\ell \in \mathrm{leaf}( \mathcal{T}_w ) }} \min\left(p_+^{(\mathcal{T}_w)}(\ell), p_-^{(\mathcal{T}_w)}(\ell) \right) \\
    &= \sum_{w \in \mathrm{child}(v)} \min\left( p^{(v)}_+(w), p^{(v)}_-(w) \right) \left(1 - \mathrm{TV}\left(p_+^{(\mathcal{T}_w}, p_-^{(\mathcal{T}_w)}\right)\right) \\
    &\geq \left( 1 - \mathrm{TV}(p_+^{(v)}, p_-^{(v)}) \right) \min_{w \in \mathrm{child}(v)} \left(1 - \mathrm{TV}\left(p_+^{(\mathcal{T}_w}, p_-^{(\mathcal{T}_w)}\right)\right). \label{eqn:tv-mult-end}
\end{align}
We can apply the induction hypothesis on $\mathcal{T}_w$ for $w \in \mathrm{child}(v)$.
This gives us
\begin{equation}
    1 - \mathrm{TV}(p_+^{(\mathcal{T}_v)}, p_-^{(\mathcal{T}_v)}) \geq \left( 1 - \mathrm{TV}(p_+^{(v)}, p_-^{(v)}) \right) \eta^{2 \epsilon t(\mathcal{T}_w)}.
\end{equation}
By definition, we have $t(\mathcal{T}_v) \geq t^{(v)} + t(\mathcal{T}_w)$ for any child node $w$ of $v$, hence
\begin{equation}
    1 - \mathrm{TV}(p_+^{(\mathcal{T}_v)}, p_-^{(\mathcal{T}_v)}) \geq \eta^{2 \epsilon \left( t(\mathcal{T}_v) - t^{(v)} \right)} \left( 1 - \mathrm{TV}(p_+^{(v)}, p_-^{(v)}) \right).
\end{equation}
From Lemma~\ref{lem:TV-one-exp} that bounds the total variation distance for a single experiment, we have
\begin{equation}\label{eqn:single-exp}
    \mathrm{TV}(p^{(v)}_+, p^{(v)}_-) \leq (1 - \eta) \min(2 \epsilon t^{(v)}, 1).
\end{equation}
Hence, we can obtain
\begin{align}
    1 - \mathrm{TV}(p_+^{(\mathcal{T}_v)}, p_-^{(\mathcal{T}_v)}) &\geq \eta^{2 \epsilon \left( t(\mathcal{T}_v) - t^{(v)} \right)} \left( 1 - (1 - \eta) \min(2 \epsilon t^{(v)}, 1) \right) \\
    &\geq \eta^{2 \epsilon \left( t(\mathcal{T}_v) - t^{(v)} \right)} \eta^{\min(2 \epsilon t^{(v)}, 1)} \\
    &\geq \eta^{2 \epsilon \left( t(\mathcal{T}_v) - t^{(v)} \right)} \eta^{2 \epsilon t^{(v)}} = \eta^{2 \epsilon t(\mathcal{T}_v)}.
\end{align}
The second inequality uses $1 - (1 - \eta) x \geq \eta^x$ for any $\eta \in (0, 0.5)$ and $x \in [0, 1]$, which follows from the convexity of $f(x) = \eta^x - 1 + (1 - \eta) x$ and the fact that $f(0) = f(1) = 0$. We have proved the inductive step.

Using induction and the fact that $t(\mathcal{T}) \leq T$, we have
\begin{equation}
    \mathrm{TV}(p_+^{(\mathcal{T})}, p_-^{(\mathcal{T})}) \leq 1 - \eta^{2 \epsilon T },
\end{equation}
which is the claimed result.
\end{proof}

\subsubsection{Lower bound from TV upper bound}
\label{sec:finalize-lowerbd}

From the reduction step in Section~\ref{sec:reduction}, for any learning algorithm with a total evolution time at most $T$ that succeeds in the learning task stated in Theorem~\ref{thm-lowerbound}, we can use the learning algorithm to successfully distinguish between $U_{\pm}(t)$ with probability at least $1 - \delta$.
By the construction of the rooted tree representation $\mathcal{T}$, after the multiple experiments, the only information the learning algorithm can access corresponds to a leaf node of the tree $\mathcal{T}$.
Hence, if the learning algorithm can distinguish between $U_{\pm}(t)$, then it can distinguish between the two probability distributions $p_+^{(\mathcal{T})}, p_-^{(\mathcal{T})}$ with probability at least $1 - \delta$.

Using LeCam's two point method, if there is an algorithm that can distinguish the two probability distributions $p_+^{(\mathcal{T})}, p_-^{(\mathcal{T})}$ with probability at least $1 - \delta$, then $1 - 2 \delta \leq \mathrm{TV}(p_+^{(\mathcal{T})}, p_-^{(\mathcal{T})})$.
Thus,
\begin{equation}
    2 \delta \geq \eta^{2 \epsilon T} \Longleftrightarrow T \geq \frac{\log(1 / 2\delta)}{2 \epsilon \log(1 / \eta)}.
\end{equation}
Recalling that $\eta \in (0, 0.5)$ is a constant close to $0$, we have
\begin{equation}
T = \Omega\left(\frac{\log(1 / \delta)}{\epsilon}\right).
\end{equation}
We have thus established Theorem~\ref{thm-lowerbound}.

\end{document}